\documentclass[thm-restate,numberwithinsect]{lipics-v2021}

\usepackage{faktor}       %
\usepackage{mathrsfs}     %
\usepackage{tikz}
\usepackage[basic]{complexity}
\usepackage{bm}
\usetikzlibrary{arrows,shapes,cd}

\let\rel\mathbf
\let\minion\mathscr
\let\complex\mathsf
\let\union\bigcup

\newcommand{\set}{\text{\normalfont \bfseries set}}
\newcommand{\cat}{\text{\normalfont \bfseries cat}}
\newcommand{\htop}{\text{\normalfont \bfseries htop}}
\newcommand{\nat}{\text{\normalfont \bfseries nat}}
\newcommand{\rels}{\text{\normalfont \bfseries rel}}

\newcommand{\ints}{\mathbb Z}

\let\Z\ints \let\R\reals \let\C\coms

\newcommand{\cF}{\mathcal{F}}

\providecommand{\X}{{\rel X}}
\providecommand{\A}{{\rel A}}
\providecommand{\B}{{\rel B}}

\providecommand{\LO}{\rel{LO}}
\providecommand{\ott}{{\rel R_3}}
\newcommand{\onein}[1]{\LO_2}

\DeclareMathOperator{\ar}{ar}
\DeclareMathOperator{\CSP}{CSP}
\DeclareMathOperator{\PCSP}{PCSP}
\DeclareMathOperator{\Pol}{pol}
\DeclareMathOperator{\conv}{conv}
\let\pol\Pol
\DeclareMathOperator{\Hom}{Hom}
\DeclareMathOperator{\mhom}{mhom}
\DeclareMathOperator{\hpoloperator}{hpol}
\newcommand{\hpol}[2][\relax]{\ifx #1\relax
  \hpoloperator(#2) \else \hpoloperator^{(#1)}(#2) \fi}
\DeclareMathOperator{\Ann}{Ann}

\DeclareMathOperator{\coker}{coker}

\DeclareMathOperator{\sk}{sk}

\newcommand{\homrel}[1]{\Hom(\ott, \rel #1)}
\newcommand{\homLO}[1]{\Hom(\ott, \LO_{#1})}
\newcommand{\affine}{\minion Z}

\newcommand{\card}[1]{\# #1}

\let\epsilon\varepsilon

\nolinenumbers
\hideLIPIcs
 
\title{Hardness of linearly ordered 4-colouring of 3-colourable 3-uniform hypergraphs}

\author{Marek Filakovský}
        {Charles University, Prague, Czech Republic}
        {marek.filakovsky@mff.cuni.cz}
        {https://orcid.org/0000-0001-5978-2623}
        {This research was supported by Charles University (project PRIMUS/21/SCI/014) and the Austrian Science Fund (FWF project P31312-N35).}

\author{Tamio-Vesa Nakajima}
        {University of Oxford, Oxford, UK}
        {tamio-vesa.nakajima@cs.ox.ac.uk}
        {https://orcid.org/0000-0003-3684-9412}
        {This research was funded by UKRI EP/X024431/1 and by a Clarendon Fund Scholarship.}

\author{Jakub Opršal}
        {University of Birmingham, Birmingham, UK}
        {j.oprsal@bham.ac.uk}
        {https://orcid.org/0000-0003-1245-3456}
        {This project has received funding from the European Union's Horizon
         2020 research and innovation programme under the Marie Skłodowska-Curie
         Grant Agreement No 101034413.}

\author{Gianluca Tasinato}
        {ISTA, Klosterneuburg, Austria}
        {gianluca.tasinato@ist.ac.at}
        {}
        {}

\author{Uli Wagner}
        {ISTA, Klosterneuburg, Austria}
        {uli@ist.ac.at}
        {https://orcid.org/0000-0002-1494-0568}
        {This research was supported by the Austrian Science Fund (FWF project P31312-N35).}

\authorrunning{M. Filakovský, T. Nakajima, J. Opršal, G. Tasinato, and U. Wagner}

\Copyright{Marek Filakovský, Tamio-Vesa Nakajima, Jakub Opršal, Gianluca Tasinato, and Uli Wagner}

\funding{For the purpose of Open Access, the authors have applied a CC-BY public copyright licence to any Author Accepted Manuscript version arising from this submission. All data is provided in full in the results section of this paper.}



\begin{document}

\maketitle

\begin{abstract}
  A linearly ordered (LO) $k$-colouring of a hypergraph is a colouring of its vertices with colours $1, \dots, k$ such that each edge contains a unique maximal colour. Deciding whether an input hypergraph admits LO $k$-colouring with a fixed number of colours is NP-complete (and in the special case of graphs, LO colouring coincides with the usual graph colouring).

  Here, we investigate the complexity of approximating the `linearly ordered chromatic number' of a hypergraph. We prove that the following promise problem is \NP-complete: Given a 3-uniform hypergraph, distinguish between the case that it is LO $3$-colourable, and the case that it is not even LO $4$-colourable. We prove this result by a combination of algebraic, topological, and combinatorial methods, building on and extending a topological approach for studying approximate graph colouring introduced by Krokhin, Opršal, Wrochna, and Živný (2023).
\end{abstract}

\section{Introduction}

Deciding whether a given finite graph is $3$-colourable (or, more generally, $k$-colourable, for a fixed $k\geq 3$) was one of the first problems shown to be \NP-complete by Karp~\cite{Karp1972}. Since then, the complexity of \emph{approximating} the chromatic number of a graph has been studied extensively. The best-known polynomial-time algorithm approximates the chromatic number of an $n$-vertex graph within a factor of $O(n \frac{(\log \log n)^2}{(\log n)^3})$ (Halldórsson~\cite{Halldorson}); conversely, it is known  that the chromatic number cannot be approximated in polynomial time within a factor of $n^{1-\varepsilon}$, for any fixed $\varepsilon>0$, unless $\mathsf{P}=\NP$ (Zuckerman \cite{Zuckerman}). However, this hardness result only applies to graphs whose chromatic number grows with the number of vertices, and the case of graphs with \emph{bounded} chromatic number is much less well understood.

Given an input graph $G$ that is promised to be $3$-colourable, what is the complexity of finding a colouring of $G$ with some larger number $k>3$ of colours? Khanna, Linial, and Safra~\cite{KLS00} showed that this problem is \NP-hard for $k=4$, and it is generally believed that the problem is \NP-hard for any constant $k$. However, surprisingly little is known, and the only improvement and best result to date, hardness for $k=5$, was obtained only relatively recently by Bulín, Krokhin, and Opršal \cite{BKO19}. On the other hand, the best polynomial-time algorithm, due to Kawarabayashi and Thorup \cite{KT17}, uses a number of colours (slightly less less than $n^{1/5}$) that depends on the number $n$ of vertices of the input graph.

More generally, it is a long-standing conjecture that finding a $k$-colouring of a $c$-colourable graph is \NP-hard for all constants $k \geq c\geq 3$, but the complexity of this \emph{approximate graph colouring} problem remains wide open. The results from \cite{BKO19} generalise to give hardness for $k=2c-1$ and all $c\geq 3$. For $c\geq 6$, this was improved by Wrochna and Živný \cite{WZ20}, who showed that it is hard to colour $c$-colourable graphs with $k=\binom{c}{\lfloor c/2\rfloor}$ colours. We remark that conditional hardness (assuming different variants of Khot's \emph{Unique Games Conjecture}) for approximate graph coloring for all $k \geq c \geq 3$  were obtained by Dinur, Mossel, and Regev \cite{DMR09}, Guruswami and Sandeep \cite{GS20}, and Braverman, Khot, Lifshitz, and Mulzer \cite{BKLM22}.

Given the slow progress on approximate graph colouring, we believe there is substantial value in developing and extending the available methods for studying this problem and related questions, and we hope that the present paper contributes to this effort. As our main result (Theorem~\ref{thm:main} below), we establish \NP-hardness of a relevant hypergraph colouring problem that falls into a general scope of \emph{promise constraint satisfaction problems}; in the process, we considerably extend a topological approach and toolkit for studying approximate colouring that was introduced by Krokhin, Opršal, Wrochna, and Živný \cite{KO19, WZ20, KOWZ23}.

Graph colouring is a special case of the \emph{constraint satisfaction problem} (\emph{CSP}), which has several different, but equivalent formulations. For us, the most relevant formulation is in terms of homomorphisms between relational structures. The starting point is the observation that finding a $k$-colouring of a graph $G$ is the same as finding a \emph{graph homomorphism} (an edge-preserving map) $G \to K_k$ where $K_k$ is the complete graph with $k$ vertices. The general formulation of the constraint satisfaction problem is then as follows (see Section~\ref{sec:PCSPs} below for more details): Fix a relational structure $\rel A$ (e.g., a graph, or a uniform hypergraph), which parametrises the problem. $\CSP(\rel A)$ is then the problem of deciding whether a given structure $\rel X$ allows a homomorphism $\rel X \to \rel A$. One of the celebrated results in the complexity theory of CSPs is the Dichotomy Theorem of Bulatov \cite{Bul17} and Zhuk \cite{Zhu20}, which asserts that for every finite relational structure $\rel A$, $\CSP(\rel A)$ is either \NP-complete, or solvable in polynomial time.

The framework of CSPs can be extended to \emph{promise constraint satisfaction problems} (\emph{PCSP}s), which include approximate graph colouring. PCSPs were first introduced by Austrin, Guruswami, and Håstad \cite{AGH17}, and the general theory of these problems was further developed by Brakensiek and Guruswami \cite{BG21}, and by Barto, Bulín, Krokhin, and Opršal \cite{BBKO21}. Formally, a PCSP is parametrised by two relational structures $\rel A$ and $\rel B$ such that there exists a homomorphism $\rel A \to \rel B$. Given an input structure $\rel X$, the goal is then to distinguish between the case that there is a homomorphism $\rel X \to \rel A$, and the case that there does not even exist a homomorphism $\rel X \to \rel B$ (these cases are distinct but not necessarily complementary, and no output is required in case neither holds); we denote this decision problem by $\PCSP(\rel A, \rel B)$. For example,  $\PCSP(K_3, K_k)$ is the problem of distinguishing, given an input graph $G$, between the case that $G$ is $3$-colourable and the case that $G$ is not $k$-colourable. This is the \emph{decision version} of the approximate graph colouring problem whose \emph{search version} we introduced above. We remark that the decision problem reduces to the search version, hence hardness of the former implies hardness of the latter.

PCSPs encapsulate a wide variety of problems, including versions of hypergraph colouring studied by Dinur, Regev, and Smyth \cite{DRS05} and Brakensiek and Guruswami~\cite{BG16}. A variant of hypergraph colouring that is closely connected to approximate graph colouring and generalises (monotone\footnote{In the present paper, we will only consider the monotone version of 1-in-3SAT, i.e., the case where clauses contain no negated variable, and we will often omit the adjective ``monotone'' in what follows.})\emph{1-in-3SAT} is \emph{linearly ordered (LO) hypergraph colouring}.
A linearly ordered $k$-colouring of a hypergraph $H$ is an assignment of the colours $[k] = \{1, \ldots, k\}$ to the vertices of $H$ such that, for every hyperedge, the maximal colour assigned to elements of that hyperedge occurs exactly once. Note that for graphs,  linearly ordered colouring is the same as graph colouring. Moreover, LO $2$-colouring of $3$-uniform hypergraphs corresponds to (monotone) 1-in-3SAT (by viewing the edges of the hypergraph as clauses). In the present paper, we focus on $3$-uniform hypergraphs; whether such a graph has an LO $k$-colouring can be expressed as $\CSP(\LO_k)$ for a specific relational structure $\LO_k$ with one ternary relation (see Section~\ref{sec:PCSPs}); in particular, 1-in-3SAT corresponds to $\CSP(\LO_2)$.

The promise version of LO hypergraph colouring was introduced by~Barto, Battistelli, and Berg \cite{BBB21}, who studied the \emph{promise 1-in-3SAT} problem. More precisely, let $\B$ be a fixed ternary structure such that there is a homomorphism $\LO_2 \to \B$. Then $\PCSP(\LO_2, \B)$ is the following decision problem: Given an instance $\rel X$ of 1-in-3SAT, distinguish between the case that $\X$ is satisfiable, and the case that there is no homomorphism $\rel X \to \B$. For structures $\B$ with three elements, Barto et al.~\cite{BBB21} obtained an almost complete dichotomy; the only remaining unresolved case is $\B = \LO_3$, i.e., the complexity of $\PCSP(\LO_2, \LO_3)$. They conjectured that this problem is \NP-hard, and more generally that $\PCSP(\LO_c, \LO_k)$ is \NP-hard for all $k \geq c \geq 2$ \cite[Conjecture 27]{BBB21}. Subsequently, the following conjecture emerged and circulated as folklore (first formally stated by Nakajima and Živný~\cite{NZ23b}): $\PCSP(\LO_2, \B)$ is either solved by the \emph{affine integer programming relaxation}, or \NP-hard (see Ciardo, Kozik, Krokhin, Nakajima, and Živný~\cite{CKK+23} for recent progress in this direction).

Promise LO hypergraph colouring was further studied by Nakajima and Živný \cite{NZ23}, who found close connections between promise LO hypergraph colouring and approximate graph colouring. In particular, they provide a polynomial time algorithm for LO-colouring 2-colourable 3-uniform hypergraphs with a superconstant number of colours, by adapting methods used for similar algorithms for approximate graph colouring, e.g., \cite{KT17}. In the other direction, \NP-hardness of $\PCSP(\LO_k, \LO_c)$ for $4 \leq k \leq c$ follows relatively easily from \NP-hardness of the approximate graph colouring $\PCSP(K_{k-1}, K_{c-1})$, as was observed by Nakajima and Živný and by Austrin (personal communications).%
\footnote{To see why, observe that $(\LO_k, \LO_c)$ promise primitive-positive defines $(K_{k-1}, K_{c-1})$; in particular, we can define $x \neq y$ by $\exists z \cdot R(z, z, x) \land R(z, z, y) \land R(x, y, z)$. We then see that if $R$ is interpreted in $\LO_k$, then the required $z$ exists if and only if $x \neq y$, as required.}

Our main result is the following, %
which cannot be obtained using these arguments.

\begin{theorem} \label{thm:main}
  $\PCSP(\LO_3, \LO_4)$ is \NP-complete.
\end{theorem}

Apart from the intrinsic interest of LO hypergraph colouring, we believe that the main contribution of this paper is on a technical level, by extending the topological approach of \cite{KOWZ23} and bringing to bear more advanced methods from algebraic topology, in particular \emph{equivariant obstruction theory}. To our knowledge, this paper is the first that uses these methods in the PCSP context; we view this as a ``proof of concept'' and believe these tools will be useful to make
further progress on approximate graph colouring and related problems.

The proof of Theorem~\ref{thm:main} has two main parts. For a natural number $n$, let $(\LO_3)^n$ be the $n$-fold power of the relational structure $\LO_3$ (see Section~\ref{sec:Polymorphisms}). In the first part of the proof, we use topological methods to show (Lemma~\ref{lem:minion-homomorphism} below) that with every homomorphism $f\colon (\LO_3)^n \to  \LO_4$, we can associate an \emph{affine map} $\chi(f)\colon \Z_3^n\to \Z_3$ (i.e., a map of the form $(x_1,\dots, x_n) \mapsto \sum_{i=1}^n \alpha_i x_i$, for some $\alpha_i\in \Z_3$ and $\sum_{i=1}^n \alpha_i \equiv 1 \pmod 3$); moreover, the assignment $f\mapsto \chi(f)$ preserves natural \emph{minor relations} that arise from maps $\pi\colon [n]\to [m]$, i.e., $\chi$ is a \emph{minion homomorphism} (see Section~\ref{sec:Polymorphisms} for the precise definitions).

In the second part of the proof, we show by combinatorial arguments that the maps $\chi(f)\colon \Z_3^n\to \Z_3$ form a very restricted subclass of affine maps: They are projections $\Z_3^n\to \Z_3$, $(x_1,\dots,x_n)\mapsto x_i$ (Corollary~\ref{cor:only-projections}). Theorem~\ref{thm:main} then follows from a hardness criterion (Theorem~\ref{thm:hardness}) obtained as part of a general algebraic theory of PCSPs \cite{BBKO21}.

In a nutshell, topology enters in the first part of the proof as follows. First, with every homomorphism $f\colon (\LO_3)^n \to  \LO_4$ we associate  a continuous map $f_\ast \colon T^n \to P^2$, where $T^n$ is the $n$-dimensional torus (the $n$-fold power of the circle $S^1$) and $P^2$ is a suitable target space that will be described in more detail later; moreover, the cyclic group $\Z_3$ naturally acts on both $T^n$ and $P^2$, and the map $f_*$ preserves these symmetries (it is \emph{equivariant}). This first step uses \emph{homomorphism complexes} (a well-known construction in topological combinatorics that goes back to the work of Lovász \cite{Lovasz-Kneser}, see Section~\ref{sec:hom-complexes}). Second, we show that equivariant continuous maps $T^n \to P^2$,
when considered up to a natural equivalence relation of symmetry-preserving continuous deformation (\emph{equivariant homotopy}), are in bijection with affine maps $\Z_3^n\to \Z_3$. This second step uses \emph{equivariant obstruction theory}.\footnote{By contrast, the topological argument in \cite{KOWZ23} required understanding maps from $T^n$ to the circle $S^1$ that preserve natural $\Z_2$-symmetries on both spaces, again up to equivariant homotopy; such maps can be classified by more elementary arguments using the \emph{fundamental group} because $S^1$ is $1$-dimensional.}

We remark that, with some additional work, our method could be extended to prove \NP-hardness of $\PCSP(\LO_k, \LO_{2k-2})$ (but as remarked above, this already follows from known hardness results for approximate graph colouring).

\section{Preliminaries}
  \label{sec:notation}

We use the notation $[n]$ for the $n$-element set $\{1, \dots, n\}$, and identify tuples $a \in A^n$ with functions $a\colon [n] \to A$, and we use the notation $a_i$ for the $i$th entry of a tuple. We denote by $1_X$ the identity function on a set $X$.

\subsection{Promise CSPs}
  \label{sec:PCSPs}

We start by recalling some fundamental notions from the theory of promise constraint satisfaction problems, following the presentation of \cite{BBKO21} and \cite{KO22}.

A \emph{relational structure} is a tuple $\A = (A; R_1^\A, \ldots, R_k^\A)$, where $A$ is a set, and $R_i^\A \subseteq A^{\ar(R_i)}$ is a relation of arity $\ar(R_i)$. The \emph{signature} of $\A$ is the tuple $(\ar(R_1), \ldots, \ar(R_k))$.
For two relational structures $\A = (A; R_1^\A, \ldots, R_k^\A)$ and $\B = (B; R_1^\B, \ldots, R_k^\B)$ with the same signature, a \emph{homomorphism} from $\A$ to $\B$, denoted $h \colon \A \to \B$, is a function $h \colon A \to B$ that preserves all relations, i.e., such that $h(a) \in R_i^\B$ for each $i \in \{1, \dots, k\}$ and $a\in R_i^\A$ where $h(a)$ denotes the componentwise application of $h$ on the elements of $a$. To express the existence of such a homomorphism, we will also use the notation $\A \to \B$.
The set of all homomorphisms from $\A$ to $\B$ is denoted by $\hom(\A, \B)$.

Our focus is on structures with a single ternary relation $R$, i.e., pairs $(A; R^\A)$ with $R^\A \subseteq A^3$. Moreover, most structures in this paper have a \emph{symmetric} relation, i.e., the relation $R^\A$ is invariant under permuting coordinates. Such structures can be also viewed as $3$-uniform hypergraphs, keeping in mind that edges of the form $(a, a, b)$ are allowed.

\begin{definition}[Promise CSP]
  Fix two relational structures such that $\A \to \B$. The \emph{promise CSP} with template $\A, \B$, denoted by $\PCSP(\A, \B)$, is a computational problem that has two versions:
  \begin{itemize}
    \item In the \emph{search} version of the problem, we are given a relational structure $\X$ with the same signature as $\A$ and $\B$, we are promised that $\X \to \A$, and we are tasked with finding a homomorphism $h\colon \X \to \B$.
    \item In the \emph{decision} version of the problem, we are given a relational structure $\X$, and we must answer \emph{Yes} if $\X \to \A$, and \emph{No} if $\X \not \to \B$. (These cases are mutually exclusive since $\A \to \B$ and homomorphisms compose.)
  \end{itemize}
\end{definition}

The decision version reduces to the search version; thus for proving the hardness of both versions of problems, it is sufficient to prove the hardness of the decision version of the problem, and in order to prove tractability of both versions, it is enough to provide an efficient algorithm for the search version.

To complete this section, we define the relational structure $\LO_k$, $k \in \mathbb{N}$, that appears in our main result. The domain of $\LO_k$ is $\{1,\ldots,k\}$, and $\LO_k$ has one ternary relation, containing precisely those triples $(a, b, c)$ which contain a \emph{unique maximum}. In other words, $(a, b, c) \in R^{\LO_k}$ if and only if $a = b < c$, $a = c < b$, $b = c < a$, or all three elements $a, b, c$ are distinct. For example, $(1, 1, 2)$ or $(1, 2, 3)$ are triples of the relation of $\LO_3$, but not $(2, 2, 1)$.

\subsection{Polymorphisms and a hardness condition}
\label{sec:Polymorphisms}

Our proof of Theorem~\ref{thm:main} uses a hardness criterion (Theorem~\ref{thm:hardness} below) obtained as part of a general algebraic theory of PCSPs developed in \cite{BBKO21}, which we will briefly review.

\begin{definition}
Given a structure $\A$, we define its $n$-fold \emph{power} to be the structure $\A^n$ with the domain $A^n$ and
  \[
    R_i^{\A^n} = \{ (a_1, \dots, a_{\ar(R_i)}) \mid (a_1(i), \dots, a_{\ar(R_i)}(i)) \in R^\A \text{ for all $i\in [n]$} \}
  \]
  for each $i$.

  An $n$-ary \emph{polymorphism} from a~structure $\A$ to a~structure $\B$ is a~homomorphism from $\A^n$ to $\B$.
  We denote the set of all polymorphisms from $\A$ to $\B$ by $\Pol(\A, \B)$, and the set of all $n$-ary polymorphisms by $\Pol^{(n)}(\A, \B)$.\footnote{Untraditionally, we use lowercase notation for polymorphisms to highlight that we are not considering any topology on them contrary to the homomorphism complexes introduced below.}
\end{definition}

Concretely, in the special case of structures with a ternary relation, a polymorphism is a~mapping $f\colon A^n \to B$ such that, for all triples $(u_1,v_1,w_1)$, \dots, $(u_n,v_n,w_n) \in R^\A$, we have
\[
  ( f(u_1, \dots, u_n), f(v_1, \dots, v_n), f(w_1, \dots, w_n)) \in R^\B.
\]

Polymorphisms are enough to describe the complexity of a PCSP up to certain $\log$-space reductions. Loosely speaking, the more complex the polymorphisms are, the easier the problem is. We will use a hardness criterion that essentially states that the problem is hard if the polymorphisms have no interesting structure. To define what do we mean by interesting structure, we have to define the notions of \emph{minor}, \emph{minion} and \emph{minion homomorphism}. 

\begin{definition}
  Fix two sets $A$ and $B$, and let $f\colon A^n \to B$, $\pi\colon [n] \to [m]$ be functions. The \emph{$\pi$-minor} of $f$ is the function $g\colon A^m \to B$ defined by $g(x) = f(x \circ \pi)$, i.e., such that
  \[
    g(x_1, \dots, x_m) = f(x_{\pi(1)}, \dots, x_{\pi(n)})
  \]
  for all $x_1, \dots, x_m \in A$. We denote the $\pi$-minor of $f$ by $f^\pi$.
\end{definition}

Abstracting from the fact that the polymorphism of any template are closed under taking minors leads to the following notion of \emph{(abstract) minions}:\footnote{Abstract minions as defined here are a generalization of so-called \emph{function minions} defined in \cite{BBKO21}; the relation between a function minion and an abstract minion is analogous to the distinction between a permutation group and a group.} %

\begin{definition}
  An \emph{(abstract) minion} $\minion M$ is a collection of sets $\minion M^{(n)}$, where $n > 0$ is an integer, and mappings $\pi^\minion M \colon \minion M^{(n)} \to \minion M^{(m)}$ where $\pi\colon [n] \to [m]$ such that
  $\pi^\minion M \circ \sigma^\minion M = (\pi\circ \sigma)^\minion M$ for each $\pi$ and $\sigma$, and $(1_{[n]})^\minion M = 1_{\minion M^{(n)}}$.\footnote{In the language of category theory, a minion is defined as a functor from the category of finite sets to the category of sets, which satisfies  a non-triviality condition: $\minion M(X) = \emptyset$ if and only $X = \emptyset$. The definition given abuses the fact that the sets $[n]$ form a skeleton of the category of finite sets.}
\end{definition}

The polymorphisms of a template $\A, \B$ form a minion $\minion M$ defined by $\minion M^{(n)} = \Pol^{(n)}(\A, \B)$, and $\pi^\minion M(f) = f^\pi$. With a slight abuse of notation, we will use the symbol $\Pol(\A, \B)$ for this minion.
Conversely, if $\minion M$ is an abstract minion, we will call $\pi^\minion M(f)$ the $\pi$-minor of $f$, and write $f^\pi$ instead of $\pi^\minion M(f)$.

An important example is the minion of projections denoted by $\minion P$. Abstractly, it can be defined by $\minion P^{(n)} = [n]$ and $\pi^\minion P = \pi$. 
Equivalently, and perhaps more concretely, $\minion P$ can also be described as follows: Given a finite set $A$ with at least two elements and integers $i\leq n$, the $i$-th $n$-ary \emph{projection} on $A$ is the function $p_i\colon A^n \to A$ defined by $p_i(x_1, \dots, x_n) = x_i$. The set of coordinate projections is closed under minors as described above and forms a minion isomorphic to $\minion P$. In particular, $\minion P$ is also isomorphic to the polymorphism minion $\Pol(\onein3, \onein3)$.

\begin{definition}
  A \emph{minion homomorphism} from a minion $\minion M$ to a minion $\minion N$ is a collection of mappings $\xi_n\colon \minion M^{(n)} \to \minion N^{(n)}$ that preserve taking minors, i.e., such that for each $\pi\colon [n] \to [m]$, $\xi_m\circ \pi^\minion M = \pi^\minion N\circ \xi_n$.
  We denote such a homomorphism simply by $\xi\colon \minion M \to \minion N$, and write $\xi(f)$ instead of $\xi_n(f)$ when the index is clear from the context.%
  \footnote{A minion homomorphism is a natural transformation between the two functors.}
\end{definition}

Using the minion $\minion P$, we can now formulate the following hardness criterion (which follows from \cite[Theorem 3.1 and Example 2.17]{BBKO21} and can also be derived from \cite[Corollary 5.2]{BBKO21}; see also Section 5.1 of that paper for more details).

\begin{theorem}[{\cite[corollary of Theorem 3.1]{BBKO21}}] \label{thm:hardness}
  For every promise template $\A, \B$ such that there is a minion homomorphism $\xi\colon \Pol(\A, \B) \to \minion P$, $\PCSP(\A, \B)$ is \NP-complete.
\end{theorem}

\subsection{Homomorphism complexes}
\label{sec:hom-complexes}

We will need a number of notions from topological combinatorics, which we will review briefly now. We refer the reader to \cite{Mat03} for a detailed and accessible introduction (see also \cite{KOWZ23}, in particular for further background on homomorphism complexes).

A \emph{(finite, abstract) simplicial complex} $\complex K$ is finite system of sets that is downward closed under inclusion, i.e., $F \subseteq G\in \complex K$ implies $F\in \complex K$. The (finite) set $V = \union \complex K$ is called the set of \emph{vertices} of $\complex K$, and the sets in $K$ are called \emph{simplices} or \emph{faces} of the simplicial complex. A \emph{simplicial map} $f\colon \complex K \to \complex L$ between simplicial complexes is a map between the vertex sets that preserves simplices, i.e., $f(F)\in \complex L$ for all $F\in \complex K$.

An important way of constructing simplicial complexes is the following: Let $P$ be a partially ordered set (poset). A \emph{chain} in $P$ is a subset  $\{p_0, \dots, p_k\} \subseteq P$ such that $p_0 < p_2 < \dots < p_k$. The set of all chains in $P$ is a simplicial complex, called the \emph{order complex} of $P$. Note that an order-preserving map between posets naturally induces a simplicial map between the corresponding order complexes.

With every simplicial complex $\complex K$, one can associate a topological space $\lvert \complex K \rvert$, called the \emph{underlying space} or \emph{geometric realization} of $\complex K$, as follows: Identify the vertex set of $\complex K$ with a set of points \emph{in general position} in a sufficiently high-dimensional Euclidean space (here, general position means that the points in $F\cup G$ are affinely independent for all $F,G\in \complex K$). Then, in particular, the convex hull $\conv(F)$ is a geometric simplex for every $F\in \complex K$, and the geometric realization can be defined as the union $\lvert \complex K \rvert=\bigcup_{F\in \complex K} \conv(F)$ of these geometric simplices (see, e.g., \cite[Lemma 1.6.2]{Mat03}). We also say that the simplicial complex $K$ is a triangulation of the space $\lvert \complex K \rvert$. Every simplicial map $f\colon \complex K \to \complex L$ between abstract simplicial complexes induces a continuous map $\lvert f \rvert\colon \lvert \complex K \rvert \to \lvert \complex L \rvert$ between their geometric realizations.  In what follows, we will often blur the distinction between a simplicial complex and its geometric realization (especially when considering properties that do not depend on a particular triangulation).

Following \cite[Section~5.9]{Mat03}, we define homomorphism complexes as order complexes of the poset of \emph{multihomomorphisms} from one structure to another.\footnote{There are several alternative definitions of homomorphism complexes that lead to topologically equivalent spaces; e.g., the definition given here is the \emph{barycentric subdivision} of the version of the homomorphism complex defined in \cite[Definition 3.3]{KOWZ23}.}

\begin{definition}
  Suppose $\A, \B$ are relational structures. A \emph{multihomomorphism} from $\A$ to $\B$ is a function $f\colon A \to 2^B \setminus \{\emptyset\}$ such that, for each relational symbol $R$ and all tuples $(u_1, \dots, u_k) \in R^\A$, we have that
  \[
    f(u_1) \times \dots \times f(u_k) \subseteq R^\B.
  \]
We denote the set of all such multihomomorphisms by $\mhom(\A, \B)$.
\end{definition}

Multihomomorphisms are partially ordered by component-wise comparison, i.e., $f \leq g$ if $f(u) \subseteq g(u)$ for all $u\in A$. They can also be composed in a natural way, i.e., if $f\in \hom(\A, \B)$ and $g\in \mhom(\B, \rel C)$, then $(g\circ f)(a) = \bigcup_{b\in f(a)} g(b)$ is a multihomomorphism from $\A$ to $\rel C$.

\begin{definition}
  Let $\A$ and $\B$ be two structures of the same signature.
  The \emph{homomorphism complex} $\Hom(\A, \B)$ is the order complex of the poset of multihomomorphisms from $\A$ to $\B$, i.e., the vertices of this simplicial complex are multihomomorphisms from $\A$ to $\B$, and faces correspond to chains $f_1 < f_2 < \dots < f_k$ of such multihomomorphisms.
\end{definition}

By the discussion above, every homomorphism $f\colon \A \to \B$ induces a simplicial map $f_*\colon \Hom(\rel C, \A) \to \Hom(\rel C, \B)$ between homomorphism complexes, and hence a continuous map between the corresponding spaces (defined on vertices by mapping a multihomomorphism $m$ to the composition $f\circ m$, and then extended linearly).

In the case of graphs, the homomorphism complex $\Hom(K_2, G)$ is commonly used\footnote{Some papers use a different complex, the so-called \emph{box complex}, that leads to a homotopically equivalent (see below) space.} to study graph colourings, including in \cite{KOWZ23}. In the present paper, we work instead with the homomorphism complex $\Hom(\ott, \A)$ where $\ott$ is the structure with $3$ elements and all \emph{rainbow tuples}, i.e., tuples $(a, b, c)$ such that $a, b$, and $c$ are pairwise distinct; this structure is a hypergraph analogue of the graph $K_2$.

Note that a homomorphism $h\colon \ott \to \A$ can be identified with a triple $(h(1), h(2), h(3)) \in R^\A$; conversely, every triple $(a, b, c) \in R^\A$ also corresponds to a homomorphism as long as $R^\A$ is symmetric.
Similarly, a multihomomorphism $m$ can be identified with a triple $(m(1), m(2), m(3))$ of subsets of $A$ such that $m(1) \times m(2) \times m(3) \subseteq R^\A$.

\subsection{Group actions}

Throughout this paper, we will work with actions of the cyclic group $\Z_3$ on various objects (relational structures, simplicial complexes, topological spaces, groups, etc.) by structure-preserving maps (homomorphisms, simplicial maps, continuous maps, etc.). Thinking of $\Z_3$ as the multiplicative group with three elements $1,\omega, \omega^2$ (with the understanding that $\omega^i\cdot \omega^j =\omega^{i+j \pmod 3}$ and $\omega^0=1$), such an action is described by describing the action of the generator $\omega$. Thus, specifying the action of $\Z_3$ on a structure $\A$ amounts to specifying a homomorphism $\omega \colon \A\to \A$ such that $\omega^3=1_A$ (hence, $\omega$ is necessarily an isomorphism; note that we are abusing notation here, writing $\omega$ both for the generator of the group and the isomorphism by which it acts). Analogously, an action of $\Z_3$ on a simplicial complex (or a topological space) is described by specifying a simplicial  isomorphism (respectively, a homeomorphism) $\omega$ of order $3$ from the complex (or space) to itself. We will mostly work with actions that are \emph{free}, which in our special case of $\Z_3$-actions simply means that $\omega$ has no fixed points.

In particular, consider the action of $\Z_3$ that acts on $\ott$ by cyclically permuting elements. This action induces an action on multihomomorphisms $h\colon \ott \to \A$ by pre-composition, and this action extends naturally to an action of $\Z_3$ on $\Hom(\ott, \A)$.\footnote{This is analogous to the action of $\Z_2$ on graph homomorphism complexes $\Hom(K_2, G)$ used, for example, in \cite{KOWZ23}.}

It is not hard to show that the action on $\Hom(\ott, \A)$ is free as long as $\A$ has no constant tuples: If a multihomomorphism $m$ is a fixed point of a non-trivial element of $\mathbb Z_3$, then $m(1) = m(2) = m(3)$, and since $m(1) \neq \emptyset$ and $m(1) \times m(2) \times m(3) \subseteq R^\A$ then $R^\A$ contains a constant tuple $(a, a, a)$ for any $a\in m(1)$. Consequently, we may observe that the action does not fix any face of the complex.

For every homomorphism $f\colon \A \to \B$, the induced simplicial map $f_*\colon \Hom(\ott, \A) \to \Hom(\ott, \B)$ (defined on vertices by mapping multihomomorphism $m$ to the composition $f\circ m$) is \emph{equivariant}; as remarked above, we will often identify $f_*$ with the corresponding continuous map between the underlying spaces.

\subsection{Homotopy}
\label{sec:homotopy}

Two continuous maps $f,g\colon X \to Y$ between topological spaces are called \emph{homotopic}, denoted $f \sim g$, if there is a continuous map $h\colon X\times [0, 1] \to Y$ such that $h(x, 0) = f(x)$ and $h(x, 1) = g(x)$; the map $h$ is called a \emph{homotopy} from $f$ to $g$. Note that a homotopy can also be through of as a family of maps $h({\cdot},t)\colon X\to Y$ that varies continuously with $t \in [0,1]$. In what follows, $X$ and $Y$ will often be given as simplicial complexes, but we emphasize that we will generally not assume that the maps (or homotopies) between them are simplicial maps.
Two spaces $X$ and $Y$ are said to be \emph{homotopy equivalent} if there are continuous maps $f \colon X \to Y$ and $g\colon Y \to X$ such that $fg \sim 1_Y$ and $gf \sim 1_X$. 

These notions naturally generalize to the setting of spaces with group actions. If $\Z_3$ acts on two spaces $X$ and $Y$ then a continuous map $f\colon X \to Y$ is ($\Z_3$-)\emph{equivariant} if $f$ preserves the action, i.e., $f\circ \omega = \omega \circ f$. Two equivariant maps $f, g\colon X \to Y$ are said to be \emph{equivariantly homotopic}, denoted by $f\sim_{\ints_3} g$, if there exists an \emph{equivariant homotopy} between them, i.e., a homotopy $h\colon X \times [0,1] \to Y$ from $f$ to $g$ such that all maps $h(\cdot,t)\colon X\to Y$ are equivariant. We denote by $[X, Y]_{\mathbb Z_3}$ the set of all equivariant homotopy classes of (equivariant) maps from $X$ to $Y$, i.e.,
\[
  [X, Y]_{\mathbb Z_3} = \{ [f] \mid f\colon X \to Y \text{ is equivariant} \}
\]
where $[f]$ denotes the set of all equivariant maps $g$ such that $f \sim_{\mathbb Z_3} g$.

\section{Overview of the proof}
  \label{sec:overview}

We give a brief overview of the proof and the core techniques used. The result is proved by a combination of topological, combinatorial, and algebraic methods. In particular, the hardness is provided by analysing the polymorphisms of the template together with a hardness criterion from \cite{BBKO21}, see Theorem~\ref{thm:hardness}. In short, our goal is to provide a minion homomorphism from the polymorphism minion $\Pol(\LO_3, \LO_4)$ to the minion of projections $\minion P$ (which is incidentally isomorphic to the polymorphism minion of 3SAT). The core of our contribution is in providing deep-enough understanding of polymorphisms of our template so that the minion homomorphism follows.
The proof has two parts: topological and combinatorial.

\subsection{Topology}

The first part builds on the topological method introduced by Krokhin, Opršal, Wrochna, and Živný \cite{KO19,WZ20,KOWZ23}. The core idea is that, similarly to approximate graph colouring, there are unreasonably many polymorphisms between the linear-ordering hypergraphs, but most of them are very similar. This means that each polymorphism contains a lot of noise but relatively little information. We use topology to remove this noise, and uncover a signal. This is done by considering the polymorphisms `up to homotopy' --- essentially claiming that the homotopy class of a polymorphism carries the information and everything else is noise.

In order to formalise this idea, we consider for each hypergraph $\rel A$ the topological space $\homrel A$. Consequently, we get that each the homomorphism $f\colon \rel A \to \rel B$ induces a continuous map $f_*\colon \homrel A \to \homrel B$. Consequently, we identify two homomorphisms $f, g$ if $f_*$ and $g_*$ are homotopic to each other. The same is then extended to polymorphisms, although this requires to overcome a few subtle technical issues. The key observation for this extension is that the power of a homomorphism complex is homotopically equivalent to a homomorphism complex of the corresponding power (see, e.g., \cite{Koz08}).

This general idea requires a refinement to avoid trivial collapses, i.e., we have to avoid the case when $f_*$ is homotopic to a constant map which is always a continuous map between topological spaces. In our case, this is avoided by keeping track of the action of $\ints_3$ on the spaces $\homrel A$ and $\homrel B$ described in the preliminaries. Consequently, we consider maps only up to $\ints_3$-equivariant homotopy (note that the map $f_*$ induced by a homomorphism is always equivariant). Further in this exposition, we will silently assume that the action is always present, and all notions are equivariant --- the formal proof below is presented with the action in mind.

At this point we sketched how to construct a map that assigns to a polymorphism $f\colon \rel A^n \to \rel B$, an equivariant continuous map $f_*\colon \homrel{A^n} \to \homrel B$. This map does not necessarily preserve minors, nevertheless, it preserves minors \emph{up to homotopy}, i.e., for each $\pi \colon [n] \to [m]$, we have that $(f^\pi)_*$ and $(f_*)^\pi$ are equivariantly homotopic (this is since $\homrel A^n$ and $\homrel {A^n}$ are only homotopically equivalent and not homeomorphic). This allows us to define a minion homomorphism between the polymorphism minion $\Pol(\rel A, \rel B)$ and the minion of `homotopy classes of continuous maps' from powers of $\homrel A$ to $\homrel B$.

\begin{definition}
  Let $X$ and $Y$ be two topological spaces with an action of $G$. The \emph{minion of homotopy classes of equivariant polymorphisms} from $X$ to $Y$ is the minion $\hpol{X, Y}$ defined by
  \[
    \hpol[n]{X, Y} = [X^n, Y]_G
  \]
  and $[f]^\pi = [f^\pi]$.
\end{definition}

Note that minors are well-defined in this minion since if $f$ and $g$ are equivariantly homotopic, then so are $f^\pi$ and $g^\pi$ for all maps $\pi$.
Hence, we have a minion homomorphism
\[
  \zeta\colon \Pol(\rel A, \rel B) \to \hpol{\homrel A, \homrel B}.
\]
This part of the proof follows \cite{KOWZ23}, namely this minion homomorphism can be constructed by following the proof of \cite[Lemma 3.22]{KOWZ23} while substituting $\ott$ for $\rel K_2$, and $\ints_3$ for $\ints_2$. We give a general categorical proof in Appendix~\ref{app:categories}.

In order to describe the minion $\hpol{\homrel A, \homrel B}$, we need to classify all homotopy classes of maps between the corresponding topological spaces. The problem of classifying maps between two spaces up to homotopy is well-studied in algebraic topology, although it can be immensely difficult, e.g., maps between spheres of dimensions $m$ and $n$ (i.e., $[S^n, S^m]$) has been classified for many pairs $m, n$, but the classification for infinitely many remaining cases is still open, and it is considered to be a central open problem in algebraic topology.
We take advantage of the topological methods developed to solve these problems. Moreover, we may simplify the matters considerably by replacing the spaces $\homLO3$ and $\homLO4$ with spaces that allow equivariant maps to and from, resp., these spaces, and are better behaved from the topological perspective. This is due to the fact that if there are equivariant maps $X' \to X$ and $Y \to Y'$, then there is a minion homomorphism
\[
  \eta\colon \hpol{X, Y} \to \hpol{X', Y'}.
\]
This minion homomorphism is defined in the same way as a minion homomorphism from $\Pol(\rel A, \rel B)$ to $\Pol(\rel A', \rel B')$ if $\rel A', \rel B'$ is a \emph{homomorphic relaxation} of $\rel A, \rel B$, i.e., if $\rel A' \to \rel A$ and $\rel B \to \rel B'$ \cite[Lemma 4.8(1)]{BBKO21}.
To substantiate our choice of $X'$ and $Y'$, let us start with describing some topological properties of the spaces $\homLO3$ and $\homLO4$.

It may be observed that the space $\homLO3$ is homotopically equivalent to the simplicial complex depicted in Fig.~\ref{fig:homlo3} where the $\ints_3$ acts on each row cyclically. We choose $X'$ so that $X' \to \homLO3$, and its powers are topologically simple but non-trivial. A natural choice is $S^1$ which can be obtained from the simplicial complex by removing all horizontal edges. The action of $\ints_3$ on the circle can be then equivalently described as a rotation by $2\pi/3$. Consequently, the powers of this space are $n$-dimensional tori $T^n$ with component-wise (diagonal) action of $\ints_3$; by definition $T^n$ is the $n$-th power of $S^1$.

\begin{figure}
  \begin{subfigure}{.32\textwidth}
      \begin{tikzpicture}[baseline=(current bounding box.center), scale = 1.6,
        every node/.style = {circle}]
        \node (a0) at (-1,0) {0};
        \node (b0) at ( 0,0) {1};
        \node (c0) at ( 1,0) {2};
        \node (a1) at (-1,1) {0'};
        \node (b1) at ( 0,1) {1'};
        \node (c1) at ( 1,1) {2'};
        \draw (a0) -- (b1) -- (c0) -- (a1) -- (b0) -- (c1) -- (a0);
        \draw (a0) -- (a1);
        \draw (b0) -- (b1);
        \draw (c0) -- (c1);
      \end{tikzpicture}
    \caption{$\homLO3$}
    \label{fig:homlo3}
  \end{subfigure}
  \begin{subfigure}{.32\textwidth}
      \begin{tikzpicture}[baseline=(current bounding box.center), scale = 1.3,
        every node/.style = {circle}]
        \node (a0) at (-30:1) {0'};
        \node (b0) at ( 90:1) {1'};
        \node (c0) at (210:1) {2'};
        \node (a1) at (150:1) {0};
        \node (b1) at (270:1) {1};
        \node (c1) at ( 30:1) {2};
        \draw [very thick] (a0) -- (b1) -- (c0) -- (a1) -- (b0) -- (c1) -- (a0);
        \draw (a0) -- (a1);
        \draw (b0) -- (b1);
        \draw (c0) -- (c1);
      \end{tikzpicture}
    \caption{$S^1$ in $\homLO3$}
  \end{subfigure}
  \begin{subfigure}{.32\textwidth}
      \begin{tikzpicture}[baseline=(current bounding box.center), scale = 1.3,
          every node/.style = {circle}]
        \node (a0) at (-1,0) {0};
        \node (b0) at ( 0,0) {1};
        \node (c0) at ( 1,0) {2};
        \node (a1) at (-1,1) {0'};
        \node (b1) at ( 0,1) {1'};
        \node (c1) at ( 1,1) {2'};
        \node (a2) at (-1,2) {0''};
        \node (b2) at ( 0,2) {1''};
        \node (c2) at ( 1,2) {2''};
        \draw (a0) -- (b1) -- (c0) -- (a1) -- (b0) -- (c1) -- (a0);
        \draw (a1) -- (b2) -- (c1) -- (a2) -- (b1) -- (c2) -- (a1);
        \draw (a0) -- (a1) -- (a2);
        \draw (b0) -- (b1) -- (b2);
        \draw (c0) -- (c1) -- (c2);
      \end{tikzpicture}
    \caption{$\homLO4$}
    \label{fig:homlo4}
  \end{subfigure}
  \caption{Some representations of spaces $\homLO3$ and $\homLO4$ up to $\ints_3$ homotopy equivalence.}
\end{figure}

The space $\homLO4$ is a bit more complicated, in particular it is not homotopically equivalent to a 1-dimensional space.
Up to equivalence, it is the order complex of the partial order depicted in Fig.~\ref{fig:homlo4} where $\ints_3$ acts on rows cyclically.\footnote{In the proof we will not need such a precise description of the space, and we will only provide an equivariant map $\homLO4 \to L_4$ where $L_4$ is the space represented by the poset in Fig.~\ref{fig:homlo4}.}
It may be observed that $\homLO4$ is simply connected, i.e., that $\pi_1(\homLO4) = 0$, and that $\pi_2(\homLO4)$ is a non-trivial group. Moreover, the action of $\ints_3$ on $\homLO4$ induces a non-trivial action of $\ints_3$ on $\pi_2(\homLO4)$. The precise group and action is described in Appendix~\ref{app:structurehomcomplexes}, nevertheless it is irrelevant for us at this point.
The space that we use to replace $\homLO4$, and denote by $P^2$, shares these two properties with $\homLO4$, and moreover $\pi_n(P^2) = 0$ for all $n > 2$. Spaces which have only one non-trivial homotopy group (and are sufficiently `nice') are called \emph{Eilenberg-MacLane spaces}, and denoted by $K(G, n)$ where $\pi_n(K(G, n)) = G$ is the only non-trivial homotopy group. These spaces are well-defined up to homotopy equivalence. They are also closely connected with cohomology: One of the core statements of obstruction theory provides a bijection $[X, K(G, n)] \simeq H^n(X; G)$ for each Abelian group $G$ and $n \geq 1$. Consequently, it is much easier to classify maps into an Eilenberg-MacLane space up to homotopy.
The space $P^2$ is in fact an Eilenberg-MacLane space $K(G, 2)$ where $G$ is a suitable group with a free action of $\ints_3$, it is chosen in such a way that it allows an equivariant homomorphism $\homLO4 \to P^2$ while allowing for much easier classification of maps into it.

Next, we prove that the minion $\hpol{S^1, P^2}$ is isomorphic to the minion $\minion Z_3$ of affine maps modulo $3$, i.e., maps $\ints_3^n \to \ints_3$ of the form $(x_1, \dots, x_n) \mapsto \sum_{i=1}^n \alpha_i x_i$ where $\alpha_1, \dots, \alpha_n \in \mathbb Z_3$ are fixed constants such that $\sum_{i=1}^n \alpha_i = 1 \pmod 3$.
We construct this minion homomorphism by classifying equivariant continuous maps from $T^n$ with the diagonal action of $\ints_3$ to $P^2$. Since we are interested in equivariant maps and equivariant homotopy, we use a version of equivariant cohomology, called \emph{Bredon cohomology}, introduced in \cite{Bre67}. For our purpose, this equivariant cohomology is defined analogously to regular cohomology except the coefficients have a $\ints_3$-action, and this action together with the action of $\ints_3$ on the space is taken into account in all computations.
The space $P^2$ has the property that for every $\ints_3$-space $X$ such that there is an equivariant map $X \to P^2$, there is a bijection $[X, P^2]_{\ints_3} \simeq H^2_{\ints_3}(X; G)$ where $G$ is the group with a $\ints_3$ action described above. Again, this is a consequence of the equivariant obstruction theory. We then compute that
\[
  H^2_{\ints_3}(T^n; G) \simeq \ints_3^{n-1},
\]
and hence observe that there are $3^{n-1}$ elements in $\hpol[n]{S^1, P^2}$. This means that $\hpol{S^1, P^2}$ and $\minion Z_3$ have the same number of elements of each arity.
To obtain the required minion isomorphism, we provide a minion homomorphism
\[
  \minion Z_3 \to [S^1, P^2]_{\ints_3},
\]
and show that it is injective. More precisely, this homomorphism is given by assigning to each affine map $f\colon \ints_3^n \to \ints_3$ (or a tuple of its coefficients), a continuous map $\mu(f)\colon T^n \to P^2$, and showing that if $f \neq g$ then $\mu(f)$ and $\mu(g)$ are not equivariantly homotopic, and that $\mu(f^\pi)$ and $\mu(f)^\pi$ are equivariantly homotopic for all $\pi$. Since both minions have the same number of elements of each arity (and this number is finite), $\mu$ is bijective, and hence a minion isomorphism. All these computations are presented in detail in Appendix~\ref{app:topology}.

The above isomorphism together with the composition of $\zeta$, $\eta$, and $\xi$ provides the following lemma.

\begin{lemma} \label{lem:minion-homomorphism}
  There is a minion homomorphism $\chi\colon \Pol(\LO_3, \LO_4) \to \affine_3$ where $\affine_3$ denotes the minion of affine maps over $\ints_3$.
\end{lemma}

This minion homomorphism is not enough to prove NP-hardness. Although we could conclude from it, for example, that $\PCSP(\LO_3, \LO_4)$ is not solved by any level of Sherali-Adams hierarchy (this is a direct consequence of \cite[Theorems 3.3 and 5.2]{DO23}). To provide hardness, we need to further analyse the image of $\chi$ which is done using combinatorial arguments.

\subsection{Combinatorics}

In the second part, which is a combinatorial argument presented in Section~\ref{sec:combinatorics}, we show that the image of $\chi$ avoids all the affine maps except of projections. This is done by analysing binary polymorphisms from $\LO_3$ to $\LO_4$.

We use the notion of \emph{reconfiguration} of homomorphisms to achieve this. Loosely speaking, a homomorphism $f$ is reconfigurable to a homomorphism $g$ if there is a path of homomorphism starting with $f$ and ending with $g$ such that neighbouring homomorphisms differ in at most one value. (For graphs and hypergraphs without tuples with repeated entries this can be taken as a definition, but with repeated entries there are two sensible notions of reconfigurations that do not necessary align.) The connection between reconfigurability and topology was described by Wrochna \cite{Wro20}, and we use these ideas to connect reconfigurability with our minion homomorphism $\xi$.

We show that any binary polymorphism $f\colon \LO_3^2 \to \LO_4$ is reconfigurable to an essentially unary polymorphism, i.e., that there is an increasing function $h\colon \LO_3 \to \LO_4$ such that $f$ is reconfigurable to the map $(x, y) \mapsto h(x)$ or to the map $(x, y) \mapsto h(y)$.
Further, we show that if $f$ and $g$ are reconfigurable to each other, then $\chi(f) = \chi(g)$. Together with the above, this means the image of $\chi_2 \colon \hpol[2]{S^1, P^2} \to \minion Z_3^{(2)}$ omits an element. More precisely, we have the following lemma where $\minion P_3$ denotes the minion of projections on a three element set (which is a subminion of $\minion Z_3$).

\begin{lemma} \label{lem:not-all}
  For each binary polymorphism $f\in \Pol^{(2)}(\LO_3, \LO_4)$, $\chi(f) \in \minion P_3^{(2)}$.
\end{lemma}

This lemma is then enough to show that the image of $\chi$ omits all affine maps except projections.

\begin{corollary} \label{cor:only-projections}
  $\chi$ is a minion homomorphism $\Pol(\LO_3, \LO_4) \to \minion P_3$.
\end{corollary}

\begin{proof}
  We show that if a subminion $\minion M\subseteq \affine_3$ contains any non-projection then it contains the map $g \colon (x, y) \mapsto 2x + 2y$. Let $f \in \minion M^{(n)}$ depends on at least 2 coordinates, and let $f(x_1, \dots, x_n) = \alpha_1 x_1 + \dots + \alpha_n x_n$.
  First assume that $\alpha_i = 2$ for some $i$. Then the binary minor given by $\pi \colon [n] \to [2]$ defined by $\pi(i) = 1$ and $\pi(j) = 2$ if $j\neq i$ is $g$ since its first coordinate is $2$ and the second is $1 - 2 = 2 \pmod 3$.
  Otherwise, we have that $\alpha_i \in \{ 0, 1 \}$ for all $i$. In particular, there are $i \neq j$ such that $\alpha_i = \alpha_k = 1$ since $f$ depends on at least 2 coordinates. Consequently, the minor defined by $\pi' \colon [n] \to [2]$ where $\pi'(i) = \pi'(j) = 1$ and $\pi'(k) = 2$ for $k \notin \{ i, j \}$ is again $g$ by a similar argument.

  Finally, the image of $\xi$ is a subminion of $\affine_3$, and since it omits $g$ and every subminion of $\affine_3$ contains $\minion P_3$, it is equal to $\minion P_3$ which yields the desired.
\end{proof}

As mentioned before, the above corollary combined with Theorem~\ref{thm:hardness} provides the main result of this paper, the $\NP$-completeness of $\PCSP(\LO_3, \LO_4)$ (Theorem~\ref{thm:main}).

\section{Combinatorics of reconfigurations}\label{sec:combinatorics}

The goal of this section is a careful combinatorial analysis of the binary polymorphisms. In particular, we will describe how the minion homomorphism $\xi\colon \Pol(\LO_3, \LO_4) \to \minion P$ acts on binary polymorphisms. This is the key to the argument that the image of $\xi$ is the projection minion and the whole of $\Pol(\ints_3)$.

We say that two polymorphisms $f, g \in \Pol^{(n)}(\LO_3, \LO_4)$ are \emph{reconfigurable} one to the other if a path between $f$ and $g$ exists within the homomorphism complex $\Hom(\LO_3^n, \LO_4)$. (Note that every polymorphism is a homomorphism $\LO_3^n \to \LO_4$, and hence a vertex of the homomorphism complex.)

We will use the following combinatorial criterion that ensures that two polymorphisms are reconfigurable to each other. The proof is subtly dependent on some properties of the structure $\LO_4$.

\begin{lemma}
  Let $\A$ be a symmetric relational structure. If $f, g\colon \A \to \LO_4$ are two homomorphisms such that $f$ and $g$ differ in exactly one value, i.e.,~there is $d \in A$ such that for all $a \in A \setminus \{ d \}$ we have $f(a) = g(a)$, then $f$ and $g$ are reconfigurable.
\end{lemma}

\begin{proof}
  We first claim that under the above assumption, the multifunction $m\colon A \to 2^{[4]}$ given by $m(a) = \{ f(a), g(a) \}$ is a multihomomorphism.
  Assume that $(a, b, c) \in R^\A$. Observe that for any $x \in A \setminus \{ d \}$ we have $f(x) = g(x)$ and hence $m(x) = \{ f(x) \} = \{ g(x) \}$. We now have cases depending on how many times $d$ appears in $\{a, b, c\}$.
  \begin{description}
      \item[$\bm{d}$ does not appear.] In this case $m(a) \times m(b) \times m(c) = \{ (f(a), f(b), f(c)) \} \subseteq R^{\LO_4}$.
      \item[$\bm{d}$ appears once.] Suppose $d = a, d \neq b, d \neq c$; then $m(a) \times m(b) \times m(c) = \{f(a), g(a)\} \times \{ f(b) \} \times \{ f(c) \} = \{ (f(a), f(b), f(c)), (g(a), g(b), g(c)) \} \subseteq R^{\LO_4}$, as $f(b) = g(b), f(c) = g(c)$.
      \item[$\bm{d}$ appears twice.] Suppose $d = a = b, d \neq c$; then as $(f(a), f(b), f(c)) = (f(d), f(d), f(c)) \in R^{\LO_4}$ and likewise $(g(d), g(d), g(c)) \in R^{\LO_4}$, we have $f(d) < f(c)$ and $g(d) < g(c) = f(c)$. Consequently, $m(a) \times m(b) \times m(c) = {\{f(d), g(d)\}}^2 \times \{ f(c) \} \subseteq R^{\LO_4}$, since every tuple has a unique maximum, namely $f(c)$.
      \item[$\bm{d}$ appears thrice.] This case (i.e.,~$d = a = b = c$) is impossible, as $\A \to \LO_4$, and thus $\A$ has no constant tuples.
  \end{description}
  Thus $m$ is a multihomomorphism in all cases.

  We can now define a path $p\colon [0, 1] \to \Hom(\LO_3, \LO_4)$ by $p(0) = f$, $p(1/2) = m$, $p(1) = g$, and extending linearly.
\end{proof}

We note, without a proof, if $f$ and $g$ are reconfigurable, then there is a sequence $f = f_0, \dots, f_k = g$ such that $f_i$ and $f_{i+1}$ differ in exactly one point.
A polymorphism $f \in \Pol^{(2)}(\LO_3, \LO_4)$ has, as its domain, the set ${[3]}^2$, and thus it can naturally be represented as a matrix:
    \[
        \begin{matrix}
            f(1, 1) & f(1, 2) & f(1, 3) \\
            f(2, 1) & f(2, 2) & f(2, 3) \\
            f(3, 1) & f(3, 2) & f(3, 3) \\
        \end{matrix}.
    \]
When we speak of ``rows'' or ``columns'' of $f$ this is what is meant.

We show the following lemma from which we will be able to derive that each binary polymorphism is reconfigurable to an essentially unary one. (Recall that a function $f\colon A^n \to B$ is essentially unary if it depends on at most one input coordinate.) The lemma is an analogue of the \emph{Trash Colour Lemma} for polymorphisms from $K_d$ to $K_{2d-2}$.

\begin{lemma}\label{lem:trashcolour}
  For each $f\in \Pol^{(2)}(\LO_3, \LO_4)$ there exists an increasing function $h\in \Pol^{(1)}(\LO_3, \LO_4)$, a coordinate $i\in \{1,2\}$, and a colour $t\in [4]$ (called \emph{trash colour}) such that
  \[
    f(x_1, x_2) \in \{ h(x_i), t \}
  \]
  for all $x_1, x_2 \in [3]$.
\end{lemma}

\begin{proof}
    Throughout we will implicitly use the fact that if $a < b$ and $c < d$ then $f(a, c) < f(b, d)$, as $((a, c), (a, c), (b, d)) \in R^{\LO_3^2}$.

    First, we claim that every colour $c\in [4]$ appears inside only one row or only one column of $f$, i.e.,~that either there is $a\in [3]$ such that $f(x, y) = c$ implies $x = a$, or there is $b \in [3]$ such that $f(x, y) = c$ implies $y = b$. For contradiction, assume that this is not the case, i.e.,~there are $x, y$ and $x', y' \in [3]$ such that $f(x, y) = f(x', y') = c$, $x \neq x'$, and $y\neq y'$. The claim is proved by case analysis as follows.
  First, observe that either $x < x'$ and $y > y'$, or $x > x'$ and $y < y'$, since otherwise $(x, y)$ and $(x', y')$ are comparable, and hence $f(x, y) \neq f(x', y')$. Since the two cases are symmetric, we may assume without loss of generality that $x < x'$ and $y > y'$.
  Furthermore, since $((x, y), (x', y'), (x, y')) \in R^{\LO_3^2}$,
  and $f(x, y) = f(x', y') = c$, we have $f(x, y') > c$. Similarly, as $x' > x, y > y'$ we have that $f(x', y) > f(x, y') > c$. This means that $c\in \{1, 2\}$. We consider each case separately.
  \begin{description}
    \item[$\bm{c = 1}$.] We claim that $x = y' = 1$ since if $x > 1$, then $f(1, y') < f(x, y) = 1$, and similarly if $y' > 1$.
        This implies that $f(1, 1) > 1$ since $((1, 1), (x, x'), (y, y')) = ((1, 1), (1, y), (x', 1)) \in R^{\LO_3^2}$ and $f(x, y) = f(x', y') = 1$. As $1 < f(1, 1) < f(2, 2) < f(3, 3) \leq 4$, we have that $f(1, 1) = 2$, $f(2, 2) = 3$, and $f(3, 3) = 4$. We now have three cases.
        \begin{description}
            \item[$\bm{y = 3}$.]
              We argue that $f(1, 2)$ has no possible value.
              First, the value $1$ is not possible since $((1, 2), (x, y), (x', y')) = ((1, 2), (1, 3), (x', 1)) \in R^{\LO_3^2}$, $f(x, y) = 1$, and $f(x', y') = 1$.
              $f(1, 2) = 2$ is not possible since $((1, 2), (1, 1), (x', y')) = ((1, 1), (1, 2), (x', 1)) \in R^{\LO_3^2}$, and $f(x', y') = 1, f(1, 1) = 2$.
              $f(1, 2) = 3$ is not possible since $((1, 2), (2, 2), (x, y)) = ((1, 2), (2, 2), (1, 3)) \in R^{\LO_3^2}$, and $f(x, y) = 1, f(2, 2) = 3$.
              Finally, $f(1, 2) < f(3, 3) = 4$, so $f(1, 2) \neq 4$.
            \item[$\bm{x' = 3}$.] Here the contradiction follows analogously to the previous case.
            \item[$\bm{x' = y = 2}$.] We consider the pair of values $f(1, 3)$ and $f(3, 1)$. First, we have $f(1, 3) > f(1, 2) = f(x, y)= 1$ and $f(3, 1) > f(2, 1) = f(x', y') = 1$. As $((1, 3), (1, 1), (x', y')) = ((1, 3), (1, 1), (2, 1)) \in R^{\LO_3^2}$ and $f(1, 1) = 2, f(x', y') = 2$ we have that $f(1, 3) \neq 2$; symmetrically $f(3, 1) \neq 2$. We also have $f(1, 3) \neq 3$ since $((1, 3), (x, y), (2, 2)) = ((1, 3), (1, 2), (2, 2)) \in R^{\LO_3^2}$ and $f(1, 2) = 1, f(2, 2) = 3$; symmetrically $f(3, 1) \neq 2$. Thus $f(1, 3) = f(3, 1) = 4$. However, then $(f(1, 2), f(1, 3), f(3, 1)) = (1, 4, 4) \not \in R^{\LO_4}$, which is not possible, as $((1, 2), (1, 3), (3, 1)) \in R^{\LO_3^2}$, which yields our contradiction.
        \end{description}
    \item[$\bm{c = 2}$.] As $f(x', y) > f(x, y') > c = 2$, we have that $f(x, y') = 3$ and $f(x', y) = 4$. Since $f(x, y') = 3$ then either $x > 1$ or $y' > 1$, otherwise $f(3, 3) > f(2, 2) > f(1, 1) = 3$ yields a contradiction. By symmetry it is enough to discuss the case $y' = 2$ and $y = 3$. Finally, we have $f(x, 1) < f(x', 2) = 2$, hence $f(x, 1) = 1$ which is in contradiction with
     \[
       (1, 2, 2) = (f(x, 1), f(x', 2), f(x, 3)) \in R^{\LO_4}.
     \]
  \end{description}
  Thus we get a contradiction in all cases, and hence each colour appears in only one row or only one column.

  We say that a colour $c\in [4]$ is of \emph{column} type if $f(x, y) = c$ implies $x = a_c$ for some fixed $a_c \in [3]$, and is of \emph{row} type if $f(x, y) = c$ implies $y = b_c$ for some $b_c \in [3]$. Note that a colour can be both row and column type, in which case we may choose either. We claim that there are at least three colours that share a type --- otherwise there are two colours of row type and two colours of column type which would leave an element of $\LO_3^2$ uncoloured.
  A similar observation also yields that there has to be three colours of the same type that cover all rows or all columns, i.e.,~such that the constants $a_c$ or $b_c$ (depending on the type) are pairwise distinct.
  Let us assume they are of the column type; the other case is symmetric. Further, we may assume that the forth colour is of the row type, since if two colours share a column, then one of the colours appears only once, and can be therefore considered to be of row type.

  We define $h(a)$ to be the colour $c$ of column type with $a_c = a$, then we have $f(x, y) \in \{h(x), t\}$ where $t$ is the colour of the row type. Finally, we argue that $h$ is increasing. This is since there are $y < y'$ with $y \neq b_t$ and $y' \neq b_t$, and consequently
  \[
    h(1) = f(1, y) < f(2, y') = h(2) = f(2, y) < f(3, y') = h(3).
  \]
  This concludes the proof of the lemma.
\end{proof}

\begin{lemma}\label{lem:reconfigurability}
    Every binary polymorphism $f \in \Pol^{(2)}(\LO_3, \LO_4)$ is reconfigurable to an essentially unary polymorphism.
\end{lemma}
\begin{proof}
    The proof relies on Lemma~\ref{lem:trashcolour}. We prove our result by induction on the number of appearances of the trash colour. The result is clear if the trash colour never appears; so assume it appears at least once. Thus suppose without loss of generality that $f(x, y) \in \{ h(x), t\}$ for some increasing $h \in \Pol^{(1)}(\LO_3, \LO_4)$, and that in particular $f(x_0, y_0) = t$. Furthermore, suppose that among all such pairs, $(x_0, y_0)$ is the one that maximises $x_0$. We claim that $f'(x, y)$, which is equal to $f(x, y)$ everywhere except that $f'(x_0, y_0) = h(x_0)$ is also a polymorphism, which gives us our inductive step.

    Consider any $((x, y), (x', y'), (x'', y'')) \in R^{\LO_3^2}$; if $(x_0, y_0) \not \in \{ (x, y), (x', y'), (x'', y'') \}$, then $(f'(x, y), f'(x', y'), f'(x'', y'')) = (f(x, y), f(x', y'), f(x'', y'')) \in R^{\LO_4}$, so assume without loss of generality that $(x'', y'') = (x_0, y_0)$. We now have two cases, depending on where the unique maximum of $(f(x, y), f(x', y'), f(x_0, y_0)) \in R^{\LO_4}$ falls.
    \begin{description}
        \item[$\bm{f(x, y)}$ is the unique maximum] In this case, $f(x, y) > f(x_0, y_0) = t$ and $f(x, y) > f(x', y')$. We must show that $f'(x_0, y_0) = h(x_0) \neq f(x, y)$. Since we know that $f(x, y) \neq t$ and thus $f(x, y) = h(x)$, and furthermore that $h$ is increasing, this is the same as showing that $x \neq x_0$. Suppose for contradiction that $x = x_0$; thus $x' > x$. If $f(x', y) = h(x') > h(x)$, then $f(x, y)$ would not be the unique maximum, so $f(x', y) = t$. This contradicts the choice of $(x_0, y_0)$, as $x' > x_0$.
        \item[$\bm{f(x', y')}$ is the unique maximum] This case is identical to the previous case.
        \item[$\bm{f(x_0, y_0)}$ is the unique maximum] It follows that $f(x, y) < t$ and $f(x', y') < t$, hence $f(x, y) = h(x)$ and $f(x', y') = h(x')$. Thus since $(x, x', x_0) \in R^{\LO_3}$ and $h$ is increasing, it follows that $(f'(x, y), f'(x', y'), f'(x_0, y_0)) = (h(x), h(x'), h(x_0)) \in R^{\LO_4}$.
    \end{description}

    Thus we see that this $f'$ is indeed a polymorphism, and contains one fewer trash colour. Thus our conclusion follows.
\end{proof}

In Figure~\ref{fig:reconfiguration}, we can see the reconfiguration graph of $\Pol^{(2)}(\LO_3, \LO_4)$. This shows how one can reconfigure all polymorphisms to essentially unary ones. In the diagram, we show a polymorphism in its matrix representation.

It can be also observed that unary polymorphisms that depend on the same coordinate are reconfigurable to each other. Moreover, since every connected component of $\Hom(\LO_3^2, \LO_4)$ contains a homomorphism, and hence a unary one, we can derive from these observation that $\Hom(\LO_3^2, \LO_4)$ has at most two connected components. In Appendix~\ref{app:topology}, we also prove that it has at least two components using topological methods.

Finally, we conclude with the statement that we actually use in the proof, which follows from well-known properties of homomorphism complexes.

\begin{lemma}
  Let $\A$, $\B$, and $\rel C$ be three structures, $G$ a group acting on $\A$, and assume that $f, g \in \hom(\B, \rel C)$ are reconfigurable. Then the induced maps $f_*, g_*\colon \Hom(\A, \B) \to \Hom(\A, \rel C)$ are $G$-homotopic.
\end{lemma}

\begin{proof}
  First, observe that the composition of multihomomorphisms as a map $\mhom(\A, \B) \to \mhom(\B, \rel C) \to \mhom(\A, \rel C)$ is monotone. This means that the composition extends linearly to a continuous map
  \[
    c: \Hom(\B, \rel C) \times \Hom(\A, \B) \to \Hom(\rel A, \rel C)
  \]
  (see also~\cite[Section 18.4.3]{Koz08}).
  Since the composition is associative, we obtain that the map $c$ is equivariant (under an action of any automorphism of $\A$ on the second coordinate).

  Finally, we have that $f_*(x) = c(f, x)$ by the definition of $f_*$, and analogously, $g_*(x) = c(g, x)$. Consequently, if $h\colon [0, 1] \to \Hom(\B, \rel C)$ is an arc connecting $f$ and $g$, i.e., such that $h(0) = f$ and $h(1) = g$, then the map $H \colon [0, 1] \times \Hom(\A, \B) \to \Hom(\A, \rel C)$ defined by
  \[
    H(t, x) = c(h(t), x)
  \]
  is a homotopy between $f_*$ and $g_*$. This $H$ is also equivariant since $c$ is equivariant.
\end{proof}

The following corollary then follows directly from the above and Lemma~\ref{lem:reconfigurability}.

\begin{corollary}
  For every binary polymorphism $f \in \Pol^{(2)}(\LO_3, \LO_4)$, the induced map $f_*\colon \homLO3^2 \to \homLO4$ is equivariantly homotopic either to the map $(x, y) \mapsto i_*(x)$, or to the map $(x, y) \mapsto i_*(y)$ where $i\colon \LO_3 \to \LO_4$ is the inclusion.
\end{corollary}

\begin{figure}
  \centerline{\includegraphics{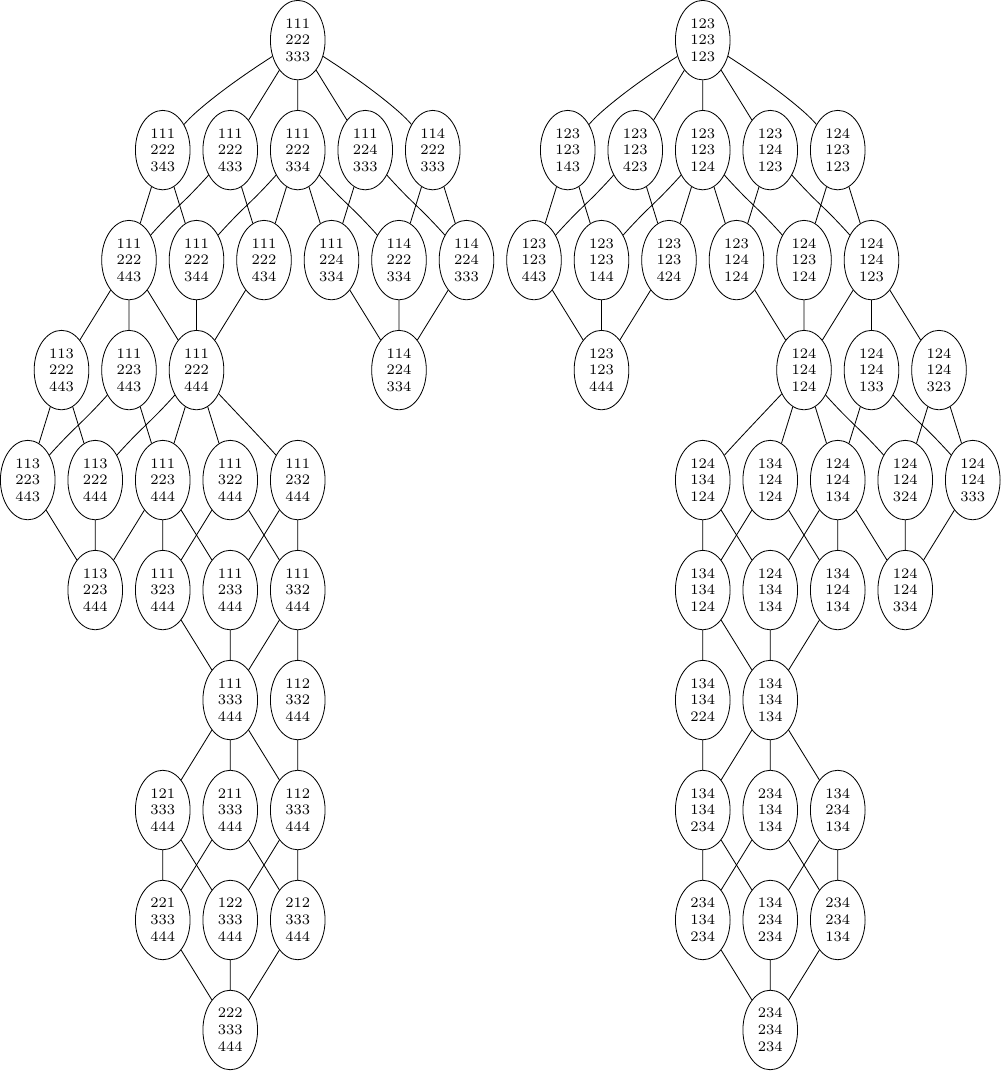}}
  \caption{Graph of reconfigurations of $\Pol^{(2)}(\LO_3, \LO_4)$.}
  \label{fig:reconfiguration}
\end{figure}

\appendix
\section{Structure of \texorpdfstring{$\homLO3$ and $\homLO4$}{homomorphism complexes}} \label{app:structurehomcomplexes}

In this section we will investigate the structure of $\Hom(\ott, \LO_3)$ and $\Hom(\ott, \LO_4)$. Luckily the first of these is easy to just draw. We can see the structure of $\Hom(\ott, \LO_3)$ in Figure~\ref{fig:hom-lo_3}. In this figure, the vertices corresponding to homomorphisms are highlighted, and labelled by the triple of the values of the homomorphism. Each of the tetrahedrons corresponds to a multihomomorphism consisting of all the edges where the unique maximum is $3$ in a fixed position $i \in [3]$. The $\Z_3$ action on $\Hom(\ott, \LO_3)$ is given by rotation.

\begin{figure}
\begin{tikzpicture}
	\fill[lipicsYellow,opacity=0.5] (canvas polar cs:radius=44mm,angle=280.0) -- (canvas polar cs:radius=32mm,angle=305.0) -- (canvas polar cs:radius=44mm,angle=320.0);
	\fill[lipicsYellow,opacity=0.5] (canvas polar cs:radius=44mm,angle=280.0) -- (canvas polar cs:radius=32mm,angle=305.0) -- (canvas polar cs:radius=56mm,angle=295.0);
	\fill[lipicsYellow,opacity=0.5] (canvas polar cs:radius=44mm,angle=280.0) -- (canvas polar cs:radius=44mm,angle=320.0) -- (canvas polar cs:radius=56mm,angle=295.0);
	\fill[lipicsYellow,opacity=0.5] (canvas polar cs:radius=32mm,angle=305.0) -- (canvas polar cs:radius=44mm,angle=320.0) -- (canvas polar cs:radius=56mm,angle=295.0);
	\fill[lipicsYellow,opacity=0.5] (canvas polar cs:radius=44mm,angle=160.0) -- (canvas polar cs:radius=44mm,angle=200.0) -- (canvas polar cs:radius=32mm,angle=185.0);
	\fill[lipicsYellow,opacity=0.5] (canvas polar cs:radius=44mm,angle=160.0) -- (canvas polar cs:radius=44mm,angle=200.0) -- (canvas polar cs:radius=56mm,angle=175.0);
	\fill[lipicsYellow,opacity=0.5] (canvas polar cs:radius=44mm,angle=160.0) -- (canvas polar cs:radius=32mm,angle=185.0) -- (canvas polar cs:radius=56mm,angle=175.0);
	\fill[lipicsYellow,opacity=0.5] (canvas polar cs:radius=44mm,angle=200.0) -- (canvas polar cs:radius=32mm,angle=185.0) -- (canvas polar cs:radius=56mm,angle=175.0);
	\fill[lipicsYellow,opacity=0.5] (canvas polar cs:radius=44mm,angle=40.0) -- (canvas polar cs:radius=32mm,angle=65.0) -- (canvas polar cs:radius=44mm,angle=80.0);
	\fill[lipicsYellow,opacity=0.5] (canvas polar cs:radius=44mm,angle=40.0) -- (canvas polar cs:radius=32mm,angle=65.0) -- (canvas polar cs:radius=56mm,angle=55.0);
	\fill[lipicsYellow,opacity=0.5] (canvas polar cs:radius=44mm,angle=40.0) -- (canvas polar cs:radius=44mm,angle=80.0) -- (canvas polar cs:radius=56mm,angle=55.0);
	\fill[lipicsYellow,opacity=0.5] (canvas polar cs:radius=32mm,angle=65.0) -- (canvas polar cs:radius=44mm,angle=80.0) -- (canvas polar cs:radius=56mm,angle=55.0);
	\node[ellipse, draw, fill=white] (112) at (canvas polar cs:radius=20mm,angle=240.0) {112};
	\node[ellipse, draw, fill=white] (113) at (canvas polar cs:radius=44mm,angle=280.0) {113};
	\node[ellipse, draw, fill=white] (121) at (canvas polar cs:radius=20mm,angle=120.0) {121};
	\node[ellipse, draw, fill=white] (123) at (canvas polar cs:radius=32mm,angle=305.0) {123};
	\node[ellipse, draw, fill=white] (131) at (canvas polar cs:radius=44mm,angle=160.0) {131};
	\node[ellipse, draw, fill=white] (132) at (canvas polar cs:radius=44mm,angle=200.0) {132};
	\node[ellipse, draw, fill=white] (211) at (canvas polar cs:radius=20mm,angle=0.0) {211};
	\node[ellipse, draw, fill=white] (213) at (canvas polar cs:radius=44mm,angle=320.0) {213};
	\node[ellipse, draw, fill=white] (223) at (canvas polar cs:radius=56mm,angle=295.0) {223};
	\node[ellipse, draw, fill=white] (231) at (canvas polar cs:radius=32mm,angle=185.0) {231};
	\node[ellipse, draw, fill=white] (232) at (canvas polar cs:radius=56mm,angle=175.0) {232};
	\node[ellipse, draw, fill=white] (311) at (canvas polar cs:radius=44mm,angle=40.0) {311};
	\node[ellipse, draw, fill=white] (312) at (canvas polar cs:radius=32mm,angle=65.0) {312};
	\node[ellipse, draw, fill=white] (321) at (canvas polar cs:radius=44mm,angle=80.0) {321};
	\node[ellipse, draw, fill=white] (322) at (canvas polar cs:radius=56mm,angle=55.0) {322};
	\draw(112) -- (113);
	\draw(112) -- (132);
	\draw(112) -- (312);
	\draw(113) -- (123);
	\draw(113) -- (213);
	\draw(113) -- (223);
	\draw(121) -- (123);
	\draw(121) -- (131);
	\draw(121) -- (321);
	\draw(123) -- (213);
	\draw[dashed] (123) -- (223);
	\draw(131) -- (132);
	\draw(131) -- (231);
	\draw(131) -- (232);
	\draw(132) -- (231);
	\draw(132) -- (232);
	\draw(211) -- (213);
	\draw(211) -- (231);
	\draw(211) -- (311);
	\draw(213) -- (223);
	\draw[dashed] (231) -- (232);
	\draw(311) -- (312);
	\draw(311) -- (321);
	\draw(311) -- (322);
	\draw(312) -- (321);
	\draw[dashed] (312) -- (322);
	\draw(321) -- (322);
	\draw[very thick] (211) -> (213) -> (113) -> (112) -> (132) -> (131) -> (121) -> (321) -> (311) -> (211);
\end{tikzpicture}
   \caption{$\Hom(\ott, \LO_3)$ where only the homomorphisms are explicitly marked.}
  \label{fig:hom-lo_3}
\end{figure}

\begin{lemma}\label{lem:homlo3}
  There exists a $\Z_3$-equivariant map $S^1 \to \Hom(\ott, \LO_3)$.
\end{lemma}

\begin{proof}
  Consider the cycle in $\Hom(\ott, \LO_3)$ given by the vertices\footnote{To be precise, every edge $f_0\to f_1$ of the cycle as it is written is not an edge of the complex, but it "goes through" the vertex corresponding to the multihomomorphism $m\colon i\mapsto f_0(i)\cup f_1(i)$. For ease of readability, these vertices have been suppressed from the notation.}
  \begin{multline*}
      (2, 1, 1) \to (2, 1, 3) \to (1, 1, 3) \to
      (1, 1, 2) \to (1, 3, 2) \\ \to (1, 3, 1) \to
      (1, 2, 1) \to (3, 2, 1) \to (3, 1, 1) \to (2, 1, 1).
  \end{multline*}
  This cycle is highlighted in Figure~\ref{fig:hom-lo_3}. Observe that it is invariant under the $\Z_3$ action. Thus we see that there is a $\Z_3$-equivariant map $S^1 \to \Hom(\ott, \LO_3)$.
\end{proof}

Now, we turn to $\Hom(\ott, \LO_4)$. Unfortunately, this is significantly more difficult to ``just see'', as it has many overlapping simplices of dimensions greater than 1.
However, luckily, for our argument it turns out we only need to find a ($\Z_3$-equivariant) map from $\Hom(\ott, \LO_4)$ to some sufficiently homotopically simple $\Z_3$ space $P^2$ (we will expand on the meaning of ``simple'' later in Section~\ref{sec:eilenberg_maclane}).
This is because the template we are looking at is $(\LO_3, \LO_4)$; thus $\LO_4$ is ``on the right''.

However, constructing a map directly from $\Hom(\ott, \LO_4)$ into $P^2$ is not entirely straightforward; instead, it is much easier to define an intermediate $\Z_3$-space $L_4$ and show that $\Hom(\ott, \LO_4) \to L_4$ and $L_4\to P^2$ equivariantly.

\begin{lemma} \label{lem:homlo-to-l}
  Let $L_4$ be the order complex of the poset depicted in Fig.~\ref{fig:homlo4}, i.e., the poset on $[3] \times \ints_3$ where $(i, \omega^\alpha) < (j, \omega^\beta)$ if $i < j$.
  Then there is a $\Z_3$-map $\Hom(\ott, \LO_4)\to L_4$.
\end{lemma}

\begin{proof}
  In order to provide the required equivariant map, we will provide an equivariant simplicial map to the barycentric subdivision of $L_4$. Moreover, since both $\homLO4$ and the barycentric subdivision of $L_4$ can be expressed as order complexes (barycentric subdivision is the order complex of the face poset of $L_4$; see \cite[Section 1.7]{Mat03}), it is enough to provide a monotone equivariant map $\phi\colon \mhom(\ott, \LO_4) \to \mathcal F$, where $\mhom(\ott, \LO_4)$ is ordered by inclusion and $\mathcal F$ is the face poset of $L_4$.

  The face poset $\mathcal{F}$ of $L_4$ can be described explicitly as follows: Its elements are those subsets $F\subseteq [3]\times \Z_3$ such that there are no two elements with the same first coordinate, i.e.,
  \[
      \mathcal{F} = \{F\subseteq [3]\times \Z_3 \mid
      \text{if $(i, \omega^\alpha)$ and $(j, \omega^\beta)$ are two distinct elements of $F$, then $i \neq j$} \}
  \]
  ordered by inclusion.

  First, we define an auxiliary function $h \colon \hom(\ott, \LO_4) \to [3] \times \ints_3$ by
  \[
    h(f) = (f(j) - 1, \omega^j)
  \]
  where $j \in [3]$ is such that $f(j)$ is the unique maximum of $f$. Using this function, we define $\phi$ by
  \[
    \phi(m) = \{h(f) \mid f\in \hom(\ott, \LO_4), f \leq m\}.
  \]
  It is straightforward to check that $\phi$ is monotone and equivariant. We will prove that it is well-defined, i.e., that $\phi(m) \in \mathcal F$.

  Observe that $F\in \cF$ if and only if $X\in \cF$ for any $X\in \binom{F}{2}$; hence, to show that $\phi(m) \in \mathcal F$, it is enough to check that $\{ h(f_0), h(f_1) \}\in \cF$ for any two distinct homomorphisms $f_0, f_1 \leq m$.
  Furthermore, $\phi(m)\notin \cF$ if and only if $f_0$ and $f_1$ share a maximum, but the maximum is attained at a different element. Say $f_0(1) = f_1(2)$ is the unique maximum, then $g \leq m$ defined by $g(1) = f_0(1)$, $g(2) = f_1(2)$, $g(3) = f_0(3)$ is not a homomorphism, which yields the contradiction.
\end{proof}

\subsection{Constructing Eilenberg-MacLane spaces}
\label{sec:eilenberg_maclane}

The goal of this section is to build explicitly our target space $P^2$.
However, the bulk of our work will be to show that the standard methods from obstruction theory can be adapted to work in our equivariant setting without fundamental alterations.
This task has been already carried out in much bigger generality by Bredon in \cite{Bre67}; the explicit construction in section \ref{sec:eilenberg} will be simply a (slightly modified) presentation of the general construction.

It is still useful to build the space $P^2$ and carry out all the computations explicitly for our context: in this way we can maintain a finer control over the topological structure of the space; control that we will use to construct explicitly the minion isomorphism in Appendix~\ref{app:topology}.

Before exploring the equivariant version of obstruction theory, it might be useful to first invest a bit of time in outlining the basic results in the non-equivariant setting.

\subsubsection{Beyond Simplicial Complexes: CW-complexes}

While the idea of simplicial complex is a powerful tool that allow for a discretization of topological spaces, for the purpose of homotopy theory they are often too rigid and unwieldy.
Instead, the concept of a CW-complex is much more natural and easy to work with.
There are multiple equivalent definitions of what a CW-complex is, we will present a constructive definition that is well suited to our use case. For a more complete and thorough presentation, we refer to \cite{Hat02}.

The basic step in constructing a CW-complex is the gluing of a disc: suppose we are given a space $Y$, then we can use a map $\phi\colon S^n\to Y$ (called attaching map) to define a new space
\[
  Z = \faktor{Y\cup D^{n+1}}{\sim}
\]
where $x\sim y$ if $x\in S^n$ and $y = \phi(x)$. The space $Z$ is usually denoted as $Y \cup_{\phi} D^{n+1}$.
Geometrically, we are identifying the boundary of the disc with the corresponding (via $\phi$) point on $Y$.

\begin{definition}
  A \emph{CW-complex (of finite type)} is the union $X$ of spaces $\sk_d X$ defined inductively as follows:
  We start with $\sk_0 X = \{p_0, \dots , p_k\}$, which is a disjoint union of points.
  Then, if we are given $\sk_d X$, we can attach a finite number of $(d+1)$-discs $D_1^{d+1}, \dots, D^{d+1}_\ell$ along some attaching maps $\phi_i\colon \partial D_i^{d+1} = S^d\to \sk_d X$ to obtain the new space $\sk_{d+1} X$ that extends $sk_d X$.
  By construction, $\sk_0 X\subseteq \sk_1 X\subseteq \dots \subseteq \sk_d X\subseteq \dots$ hence it is possible to define the union
  \[
    X = \bigcup_{d = 0}^\infty \sk_d X
  \]

  For any $d\geq 0$, the space $\sk_d X$ is $d$-dimensional and it is called \emph{$d$-skeleton} of $X$.
  A CW-complex $X$ is $n$-dimensional if there is $n\geq 0$ such that $\sk_i X = \sk_{i+1} X$ for any $i\geq n$ (i.e., no cells of dimension higher than $n$ were used in ``building'' $X$).
\end{definition}

\begin{example}
  \begin{itemize}
    \item $S^n$ is a CW-complex by gluing a $n$-disc to a point, or by gluing, for all $k\leq n$ two $k$-discs to an ``equatorial'' $S^{k-1}$. %
    \item Every simplicial complex is also a CW-complex: any $k$-simplex in a simplicial complex $\complex L$ can be seen as glued to its faces via linear maps.
    \item If $X$ and $Y$ are CW-complexes, then $X\times Y$ is again a CW-complex.
  \end{itemize}
\end{example}

\begin{definition}
  A map $f\colon X\to Y$ between CW-complexes is \emph{cellular} if for all $d\geq 0$
  \[
    f(\sk_d X)\subseteq \sk_d Y.
  \]
\end{definition}

CW-complexes are the natural way to ``discretize'' homotopy theory:  this is because they are much more flexible than simplicial complexes while still giving an manageable description of the available maps thanks to the Cellular Approximation Theorem \cite[Theorem 4.8]{Hat02}.

\begin{theorem}[Cellular Approximation Theorem]
  Let $X$, $Y$ be CW-complex and $f\colon X\to Y$ a continuous map. Then there is a cellular map $\tilde f\colon X\to Y$ that is homotopic to $f$.

  Furthermore, if there is a subcomplex $A$ in $X$ where $f$ is already cellular then the homotopy can be taken to be stationary on $A$, i.e., $h(t, x) = f(x)$ for any $x\in A$, $0\leq t\leq 1$.
\end{theorem}

This means that, as far as homotopy theory is concerned, every map between CW-complex is a cellular map.

\begin{remark}
  In a CW-complex, the $d$-skeleton gives generators and the $d+1$ skeleton gives relations for the $d$-homotopy group.
  Formally, let $X$ a connected CW-complex with base point $x_0\in X$, pick $d\geq 0$ and let $\iota_d: \sk_d X \hookrightarrow X$. Then $(\iota_d)_*:\pi_d(\sk_d X, x_0)\to \pi_d(X, x_0)$ is surjective while $(\iota_{d+1})_*:\pi_d(\sk_{d+1}X, x_0)\to \pi_d(X, x_0)$ is an isomorphism.
\end{remark}

\subsubsection{Classical versus equivariant obstruction theory}

One of the most important concept at the core of obstruction theory is the definition of Eilenberg-MacLane spaces, spaces that have exactly one non-zero homotopy group in a prescribed dimension and isomorphism class.

\begin{definition}
  Let $G$ be a group and $n\geq 1$ (assume that $G$ is Abelian if $n\geq 2$). Then a space is an Eilenberg-MacLane space $K(G, n)$ with base point $x_0\in K(G, n)$ if it is connected and
  \[
      \pi_k(K(G, n), x_0) \cong \begin{cases}
        G & \text{if $k = n$,}\\
        0 & \text{otherwise.}
      \end{cases}
  \]
\end{definition}

It is a standard homotopic argument to show that there is an Eilenberg-MacLane space for any possible choice of $G$ and $n$ (see, e.g., \cite[Section 4.2]{Hat02}).
The key property that we are interested in is the following well-known representability result (see, e.g., \cite[Theorem 4.57]{Hat02}).

\begin{theorem}
  Let $X$ be a CW complex with base point $x_0\in X$ and $G$ an Abelian group.
  There is a natural isomorphism\footnote{To be precise, $K({-}, n)$ is a functor from (Abelian) groups into pointed topological spaces up to base point preserving homotopy equivalences and the stated isomorphism is natural both in $X$ and $G$.} between base point preserving maps up to base point preserving homotopies and cohomology of the space
  \[
      [X, K(G, n)]_{*} \cong H^n(X; G).
  \]
\end{theorem}

Therefore, it is natural to try and use an equivariant version of Eilenberg-MacLane spaces to obtain a similar isomorphism for equivariant homotopy classes and a suitable version of equivariant cohomology classes.

The main difference in the equivariant setting is that there is no map that naturally plays the role of ``0'' as the constant map does for the classical setting; as a result, to obtain the desired isomorphism it is necessary to first choose a ``reference'' map (see \cite[Theorem II.3.17]{Die87}):\footnote{This theorem is shown also in \cite{Bre67} under slightly different topological hypothesis that makes it awkward to use in our setting.}

\begin{restatable}{theorem}{obstruction}\label{thm:obstruction}
  Let $X$ be a CW-complex with a free $\Z_3$ action, let $G$ be an Abelian group with an action of $\Z_3$ and let $K(G, n)$ be an Eilenberg-MacLane space with a free $\Z_3$ action that induces on $\pi_n(K(G, n))\cong G$ the same $\Z_3$-structure.

  Then, given an equivariant map $f: X \to K(G, n)$, there is an isomorphism
  \[
      \phi_f\colon \left[X, K(G, n)\right]_{\Z_3}\to H_{\Z_3}^n(X; G)
  \]
  where $H_{\Z_3}^n\left(X; G\right)$ is the Bredon cohomology, defined below, in dimension $n$ ($\phi_f$ in general depends on the equivariant homotopy class of $f$).
\end{restatable}

It is now clear what we want to achieve: we want to construct an equivariant Eilenberg-MacLane space $P^2 = K(G, n)$, for some $n\geq 1$ and Abelian group $G$, with a free $\Z_3$-action such that $L_4\to K(G; n)$ equivariantly, and compute the Bredon cohomology module $H_{\Z_3}^n(T^r; G)$ for all $r\geq 1$.

\subsubsection{\texorpdfstring{Constructing $P^2$}{Constructing the space}}\label{sec:eilenberg}

We will build the space $P^2$ by inductively gluing higher dimensional cells on the lower dimensional skeleta.

We start with a $2$-dimensional CW-complex $Y_2$ with three vertices $\{v_0, v_1, v_2\}$, three edges $\{e_0, e_1, e_2\}$ and three $2$-discs $\{d_0, d_1, d_2\}$ glued according to the following maps:
\begin{align*}
  \phi_i\colon \partial e_i &\rightarrow \sk_0(Y_2)\\
  \phi_i(\partial e_i) &= v_{i+1} - v_i
\end{align*}
and
\begin{align*}
  \psi_i\colon \partial d_i &\rightarrow \sk_1(Y_2)\\
  \psi_1(\partial d_i) &= e_i + e_{i+1} + e_{i+2}.
\end{align*}

Geometrically, $Y_2$ is a sphere $S^2$ with a disc glued along an equator. The $\Z_3$-action is also geometrically pretty simple: on the equator ($S^1 = \sk_1Y_2$) the generator $\omega$ acts as a rotation by $2\pi/3$, while it cyclically exchanges the three discs around.

Formally, if $R\colon D^2\to D^2$ is the rotation by $2\pi/3$ around the origin, then the generator $\omega$ acts cellularly as (the addition in the indexes is considered modulo $3$):
\begin{align*}
  \omega\cdot v_i &= v_{i+1}\\
  \omega\cdot e_i &= e_{i+1}\\
  \omega\cdot d_i &= R(d_{i+1})
\end{align*}

\begin{lemma}
  The space $Y_2$ is simply connected and the action on it is free.
\end{lemma}

\begin{proof}
  The action is cellular and no cell is mapped to itself, therefore it is free.
  Up to homotopy, $Y_2$ is equivalent to $S^2 \vee_{v_0} S^2$, the space obtained by glueing two spheres at one point (the equivalence contracts the ``middle'' disc of $Y_2$ to the vertex $v_0$).

  By the Cellular Approximation Theorem, every map $\gamma\colon (S^1, p)\to (Y_2, v_0)\sim (S^2\vee_{v_0}S^2, v_0)$ can be assumed to map in $\sk_1(S^2\vee_{v_0}S^2) = \{v_0\}$ and thus nullhomotopic.
\end{proof}

\begin{lemma} \label{lem:l-to-y}
  There is an equivariant map $L_4\to Y_2$.
\end{lemma}

\begin{proof}
  We want to build the map by induction on the skeleton of $L_4$, if there were no actions involved, this would be a straightforward obstruction theoretic argument.
  However, since the action on $L_4$ is free, it is enough to take care of one representative per orbit and ``translate'' along the orbit.
  We will carry out the argument explicitly since it exemplifies well the standard obstruction theoretic arguments.

  We start by picking a representative $V_i$ in every orbit of vertices of $L_4$ and define $f_0(V_i) = v_0$. Then, by equivariance there is a unique way to extend $f_0$ to the whole vertex set:
  given $V\in \sk_0 L_4$, there is a unique $\omega^\alpha$ and unique $V_i$ such that $\omega^\alpha\cdot V_i = V$, hence define $f_0(V) = v_\alpha$. The map is well defined on $\sk_0 L_4$ and it is also equivariant.

  Now, pick a representative $E_j$ for every orbit of edges in $L_4$ and let $V_j, W_j$ be its endpoints. The vertices $f_0(V_j)$ and $f_0(W_j)$ are already defined but, since $Y_2$ is connected, there is a path $\gamma_j$ connecting $f_0(V_j)$ and $f_0(W_j)$; hence we can map the edge $E_j$ along the chosen curve $\gamma_j$.
  Again, by equivariance we have a unique way to extend $f_1$ to $\sk_1 L_4$: if $\omega^\alpha E_j$ is an edge, its endpoints are $\omega^\alpha \cdot V_j, \omega^\alpha \cdot W_j$ and $\omega^\alpha \cdot \gamma_j$ is a path connecting their images; thus we can map the edge $\omega^\alpha \cdot E_j$ to the path $\omega^\alpha \cdot \gamma_j$.

  Finally, pick a representative $D_k$ for every orbit of triangles in $L_4$. The map on the boundary $\partial D_k \cong S^1$ is already defined so $f_1|_{\partial D_k}\colon S^1\to Y_2$ is a loop in $Y_2$. Since $Y_2$ is simply connected, $f_1|_{\partial D_k}$ extends to a map defined on the whole disc $f_2\colon D_k\cong D^2\to Y_2$.
  Once again, completing the definition of $f_2$ along the orbits in the same way as before we obtain an equivariant map $f = f_2\colon L_4 \to Y_2$, as desired.
\end{proof}

 We can now build the rest of the space $P^2$ by inductively glueing $(i+1)$-discs to the $i$-skeleton $Y_i$ obtained at the previous step.
 First, we fix $Y_3 = Y_2$, i.e., we don't glue any $3$-discs.
 We work out the details in the first step of the construction to exemplify the procedure.

 The goal is to glues cells to representatives of generators of $\pi_3(Y_3)$ so that they become nullhomotopic in the new space.
 By a result of Serre \cite{Ser53}, all the homotopy groups of a CW-complex (of finite type) are finitely generated so we can fix representatives $\gamma_0, \dots, \gamma_{d_4}$ for a set of generators $[\gamma_0], \dots, [\gamma_{d_4}]$ of $\pi_3(Y_3)$.

 We can glue three $4$-dimensional discs $\{D^4_{0, i}, D^4_{1, i}, D^4_{2, i}\}$ per generator $\gamma_i$ using as attaching maps $\gamma_i, \omega\cdot\gamma_i$ and $\omega^2\cdot \gamma_i$; the action extends naturally and so we obtain a new $\Z_3$-CW-complex
  \[
    Y_4 = \bigcup_{0\leq i \leq d_4} \left(Y_3\cup_{\gamma_i} D^4_{0, i}\cup_{\omega\cdot\gamma_i} D^4_{1, i} \cup_{\omega^2\cdot\gamma_i}  D^4_{2, i}\right)
  \]

  Since we glued disc of dimension $4$, $\sk_3 Y_4 = Y_3$; by cellular approximation, $\pi_i(Y_4)\cong \pi_i(Y_3)$ for $i\leq 2$, while $\iota_*: \pi_3(Y_3)\to \pi_3(Y_4)$ is surjective.
  For any $[\gamma]\in \pi_3(Y_3)$, $[\gamma] = \sum n_i [\gamma_i]$ hence $\iota_*([\gamma])=\sum n_i\iota_*([\gamma_i]) = 0$. Therefore, $\pi_3(Y_4) = 0$. What is more, $Y_4$ is still a finite CW-complex hence Serre theorem still applies and $\pi_\ell(Y_4)$ remains finitely generated for any $\ell$.
  Finally, the inclusion $\iota\colon Y_3\hookrightarrow Y_4$ is $\Z_3$-equivariant by construction and the action on $Y_4$ remains consistent with the action on the $3$-skeleton.

  Analogously, given $Y_k$ we can build the space $Y_{k+1}$ by fixing representatives $\gamma_0, \dots, \gamma_{d_k}$ for a set of generators of $\pi_k(Y_k)$ and gluing three $(k+1)$-discs $\{D^{k+1}_{0, i}, D^{k+1}_{1, i}, D^{k+1}_{2, i}\}$ per generator in a similar fashion:
\[
  Y_{k+1} = \bigcup_{ 0\leq i \leq d_k} \left(Y_k\cup_\gamma D^{k+1}\cup_{\omega\cdot\gamma} D^{k+1}\cup_{\omega^2\cdot\gamma} D^{k+1}\right)
\]

By cellular approximation, we have not changed anything in the homotopy groups in dimension less than $k$ and $\iota_*\colon \pi_k(Y_k)\to \pi_k(Y_{k+1})$ is surjective. This again implies that $\pi_k(Y_{k+1}) = 0$ and by Serre's Theorem every homotopy group remains finitely generated.

Finally, we define $P^2$ as the union of $Y_k$'s.

\begin{definition} \label{def:p2}
  Let $Y_k$ be as above. Using $Y_k\subseteq Y_{k+1}$ for all $k$, we define $P^2$ by
  \[
    P^2 = \bigcup_{k \geq 2} Y_k.
  \]
\end{definition}

By construction, $P^2$ is a CW-complex with $\sk_{k+1} P^2 = Y_{k+1}$ and consequently, for all $k \geq 1$,
\[
  \pi_k(P^2) \cong
  \begin{cases}
    \pi_2(Y_3) & \text{ if $k = 2$, and }\\
    0 & \text{ otherwise. }
  \end{cases}
  \]
We have constructed an Eilenberg-MacLane space $K(\pi_2(Y_2), 2)$ but $P^2$ has some additional properties that we can use (note that $Y_2 = Y_3$, and hence $\pi_2(Y_2) = \pi_2(Y_3)$).
In particular, $P^2$ has a free cellular $\Z_3$-action that acts transitively on vertices, edges and $2$-cells. Furthermore, $P^2$ does not have any $3$-cell in its CW-decomposition, which we will use later. We also have the following.

\begin{lemma} \label{lem:homlo4}
  There is an equivariant map $\homLO4 \to P^2$.
\end{lemma}

\begin{proof}
  We have $\homLO4 \to L_4 \to Y_2 \to P^2$ where the first two maps follow from Lemmata~\ref{lem:homlo-to-l} and \ref{lem:l-to-y}, and the last one from the definition of $P^2$ (Definition~\ref{def:p2}).
\end{proof}

\section{\texorpdfstring{Computation of $[T^n, P^2]_{\ints_3}$}{Computation of the cohomology classes of maps}}\label{app:topology}

The goal of this appendix is to describe the minion $\hpol{S^1, P^2}$. In particular, we classify classes of equivariant maps $[T^n, P^2]_{\ints_3}$ by providing the following lemma which is a key step in proving Lemma~\ref{lem:minion-homomorphism}.

\begin{lemma} \label{thm:affine_minion}
  There is a minion isomorphism
  \[
    \mu\colon \affine_3  \to \hpol{S^1, P^2}.
  \]
\end{lemma}

Recall that $\affine_3$ is the minion of affine maps $f\colon \ints_3^n \to \ints_3$ such $f(1, \dots, 1) = 1$. Since each such map is uniquely described by its coefficients we get that there are in bijections with the $n$-tuples $(\alpha_1, \dots, \alpha_n) \in \ints_3^n$ such that $\sum_i \alpha_i = 1$. In this appendix, we view $\affine_3$ as an abstract minion whose elements are tuples of the form above, in particular, we have $\affine_3^{(n)} \subseteq \ints_3^n$.

The proof of Lemma~\ref{thm:affine_minion}, while straightforward, it is rather technical so it might be useful to outline the structure of the argument.
We start with describing a family of equivariant maps that are expressed as multivariate monomials over $\mathbb C$. In particular observe that each monomial restricts to a map $S^1 \to S^1$ (where $S^1 \subset \mathbb C$ is the unit circle), and moreover the monomials whose total degree is congruent to 1 modulo $3$ induce $\mathbb Z_3$-equivariant map.
These maps are then composed with a fixed embedding $S^1 \to P^2$ to obtain an equivariant map $T^n = \left(S^1\right)^n \to P^2$ (the choice of a concrete embedding is irrelevant since $P^2$ is simply connected).

\begin{definition}
  Fix an embedding $\iota\colon S^1 \hookrightarrow P^2$ (e.g., the inclusion of the $1$-skeleton), and let $n\geq 1$.
  We assigns to each tuple $\alpha \in \ints^n$ with $\sum_i \alpha_i = 1 \pmod 3$, a \emph{monomial map} $m_\alpha\colon T^n \to P^2$, defined by
  \[
    m_\alpha(z_1, \dots, z_n) = z_1^{\alpha_1}\cdots z_n^{\alpha_n}
  \]
\end{definition}

It is straightforward to check that each $m_\alpha$ is equivariant, indeed
\begin{multline*}
    m_\alpha\left(\omega\cdot(z_1, \dots, z_n)\right)
    = m_\alpha\left(\omega z_1, \dots, \omega z_n\right)
    = \iota(\textstyle \prod_j \omega^{\alpha_j}z_j^{\alpha_j})
    = \iota(\omega \textstyle \prod_j z_j^{\alpha_j}) \\
    = \omega \iota(\textstyle \prod_j z_j^{\alpha_j})
    = \omega\cdot m_\alpha(z_1, \dots, z_n)
\end{multline*}
where the third equality used that $\omega^{3k+1} = \omega$.

\begin{remark}
  It is not hard to observe that the assignment $\alpha \mapsto m_\alpha$ preserves minors when $\alpha$ is interpreted as a function $f\colon \ints^n \to \ints$. While we implicitly use this minion homomorphism, this is not the minion homomorphism we are looking for --- importantly, $\affine_3$ is \emph{not} a subminion of the minion of tuples $\alpha \in \mathbb Z^n$ with $\sum \alpha_i = 1 \pmod 3$ since, e.g., the unary minor of $(2, 2)$ disagrees in the two minions.
\end{remark}

We then prove that (in Section~\ref{sec:all_distinct}), for $\alpha, \beta \in \ints^n$, as long as $\alpha_i \neq \beta_i \pmod 3$, then $m_\alpha$ and $m_\beta$ are not equivariantly homotopic.
This then allow us to define an injective mapping $\mu\colon \affine_3 \to \hpol{S^1, P^2}$ --- we view each element $\alpha \in \affine_3^{(n)}$ as a tuple in $\ints$ with $\alpha_i \in \{0, 1, 2\}$.

In Section~\ref{sec:no_more}, we show that any equivariant map $T^n \to P^2$ is homotopic to some monomial map $m_\alpha$ for $\alpha \in \affine_3^{(n)}$. This is done by showing that $\#[T^n, P^2]_{\Z_3} = 3^{n-1}$ which we show via obstruction theory --- the bijection is given by a bijection between $[T^n, P^2]_{\ints_3}$ and the Bredon cohomology $H^2_{\ints_3}(T^n, G)$ (here use that $P^2$ is an Eilenberg-MacLane space), and a computation that gives that $H^2_{\ints_3}(T^n, G) \simeq \ints_3^{n-1}$.
Consequently, we get that $\mu$ is a bijection (or more precisely, that $\mu_n\colon \affine^{(n)} \to [T^n, P^2]_{\ints_3}$ is a bijection for all $n$).

Finally, we show that $\mu$ is a minion homomorphism, i.e., that it preserves minors. This is done by observing that the \emph{inverse} of $\mu$ (which is in fact described in Section~\ref{sec:all_distinct}) preserves minors. Similarly to homomorphisms of algebraic structures (e.g., groups), an inverse of a bijective minion homomorphism is a homomorphism as well.

\subsection{Monomial maps are all homotopically distinct}
  \label{sec:all_distinct}

Here, we prove the following:

\begin{lemma} \label{lem:B.3} \label{lem:all_distinct}
  Let $\alpha, \beta \in \affine_3^{(n)}$. If $m_\alpha$ is equivariantly homotopic to $m_\beta$, then $\alpha = \beta$.
\end{lemma}

The proof of the lemma is build on the following informal geometric intuition.
By cellular approximation, a homotopy between two maps $f_0, f_1\colon T^n\to P^2$ ``drags'' around the image of any of the coordinate circles (i.e., the image of $c_i\colon S^1 \hookrightarrow T^n$ defined by $c_i(x) = b$ where $b_i = x$ and $b_j$ is constant for each $j\neq i$) of the torus along some $2$-disc in $P^2$. However since it has to be an equivariant homotopy, it has to drag in the same way the rest of the orbit of the coordinate cycle equivariantly along the other discs in the image. Therefore the number of times a coordinate cycle wraps around can only change by a multiple of $3$.

The formalization of this idea, however, requires some work. In particular, we need a way to make the idea of a ``degree'' of a coordinate cycle precise:

\begin{definition}\label{def:degree}
    Let $e^j\in C^1(P^2)$ and $d^j\in C^2(P^2)$ be the dual of $e_j$ and $d_j$ respectively (i.e., $e^j(e_i) = 0$ if $i\neq j$, and $e^j(e_i) = 1$ if $i = j$). Denote by $x_i\in C_1(T^n)$ the $i$-th coordinate cycle in $T^n$ and $b_i, B_i\in C_2(T^n)$ fillings for $x_i - \omega\cdot x_i$ and $x_i - \omega^2\cdot x_i$ respectively.

    Then, the $i$-degree of an equivariant map $f\colon T^n\to P^2$ is
    \[
        \deg_i(f) = \left(f^*(e^0)(x_i) + f^*(d^0)(b_i) + f^*(d^0)(B_i)\right) \bmod 3
    \]
\end{definition}

A key observation is that this quantity does not change along equivariant homotopies:

\begin{lemma} \label{lem:B.4}
    Let $f_0, f_1: T^n\to P^2$ equivariant maps homotopic via an equivariant homotopy $h$. Then for all $i\in [n]$, $\deg_i(f_0) = \deg_i(f_1)$.
\end{lemma}

\begin{proof}
    On the cochains of dimension $1$ we have:
    \[
        f_0^*(e_0)(x_i) - f_1^*(e_0)(x_i) = (\delta h^*(e^0))(x_i) + (h^*\delta(e^0))(x_i) = 0 + h(d^0+d^1+d^2)(x_i)
    \]
    where the second equality is obtained by using the fact that $\partial x_i = 0$ and $\delta e^0 = d^0+d^1+d^2$.

    On the $2$-cochains we have:
    \begin{align*}
        f_0^*(d^0)(b_i) - f_1^*(d^0)(b_i) &= (\delta h^*(d^0))(b_i) + (h^*\delta(d^0))(b_i) = h^*(d^0)(x_i - \omega \cdot x_i) + 0\\
        f_0^*(d^0)(B_i) - f_1^*(d^0)(B_i) &= (\delta h^*(d^0))(B_i) + (h^*\delta(d^0))(B_i) = h^*(d^0)(x_i - \omega^2\cdot x_i) + 0
    \end{align*}
    where we use that $\partial b_i = x_i - \omega \cdot x_i$, $\partial B_i = x_i - \omega^2 \cdot x_i$ and $\delta d^0 = 0$.

    However, since the homotopy is equivariant, it has to commute with the action and thus $h^*(d^0)(\omega \cdot x_i) = h^*(\omega \cdot d^0)(x_i)$.
    Summing everything together, using this fact and that $\omega^i \cdot d^0 = d^i$, we obtain that
    \begin{align*}
        \deg_i(f_0) - \deg_i(f_1)
        &= h(d^0+d^1+d^2)(x_i)
          + h^*(d^0)(x_i - \omega \cdot x_i)
          + h^*(d^0)(x_i - \omega^2\cdot x_i) \\
        &= h(d^0)\left((x_i + \omega \cdot x_i + \omega^2 \cdot x_i)
          + (x_i - \omega \cdot x_i) + (x_i - \omega^2 \cdot x_i)\right) \\
        &= h^*(d^0)(3x_i) = 3 h^*(d^0)(x_i) = 0 \pmod 3. \qedhere
    \end{align*}
\end{proof}

Furthermore, we can show the following which will be used to prove that the mapping $\mu$ is injective, and also (later) to show that its ``inverse'' is a minion homomorphism.

\begin{lemma}\label{lem:monomial_degrees}
    Let $\alpha \in \Z^n$ such that $\sum_i \alpha_i = 1 \pmod 3$. Then
    \[
        \deg_i(m_\alpha) = \alpha_i \bmod 3
    \]
\end{lemma}

\begin{proof}
    Since the image of $m_\alpha$ is contained in the $1$-skeleton, $m_\alpha^*(d^0)\equiv 0$.
    Moreover, $(m_\alpha)_*(x_i) = \alpha_i e_0 +\alpha_i e_1 +\alpha_i e_2$, hence $e^0((m_\alpha)_*(x_i)) = \alpha_i$ and $\deg_i(m_\alpha) = \alpha_i \bmod 3$.
\end{proof}

We can now conclude the proof of Lemma~\ref{lem:B.3} which follows immediately from the above.

\begin{proof} [Proof of Lemma~\ref{lem:B.3}]
  If $\alpha \neq \beta$, then they differ in at least one coordinate, i.e.,  $\alpha_i \neq \beta_i$ for some $i$. Then $\deg_i(m_\alpha)\neq \deg_i(m_\beta)$ by Lemma~\ref{lem:monomial_degrees}, and therefore $m_\alpha$ and $m_\beta$ are not equivariantly homotopic by Lemma~\ref{lem:B.4}.
\end{proof}

Before, we progress further, let us discuss one more consequence of Lemma~\ref{lem:B.3}, namely, the following.

\begin{lemma} \label{lem:gamma_is_homomorphism}
  The mapping $\gamma\colon \hpol{S^1, P^2} \to \affine_3$ defined by
  \[
    [f] = (\deg_1(f), \dots, \deg_n(f))
  \]
  satisfies, for each $\alpha \in \mathbb Z^n$ with $\sum_i \alpha_i = 1$, and $\pi \colon [n] \to [m]$,
  \[
    \gamma([m_\alpha]^\pi) = \gamma([m_\alpha])^\pi,
  \]
  i.e., it preserves minors when restricted to classes of monomial maps.
\end{lemma}

\begin{proof}
  Note that $\gamma$ is well-defined since the degrees do note depend on the choice of representative (Lemma~\ref{lem:B.4}).
  Further, using Lemma~\ref{lem:monomial_degrees}, we have that
  \[
    \gamma([m_\alpha]) = (\alpha_1 \bmod 3, \dots, \alpha_n \bmod 3)
  \]
  and hence $\gamma([m_\alpha])^\pi = \beta$ where $\beta_j = (\sum_{i\in \pi^{-1}(j)} \alpha_i) \bmod 3$.
  Furthermore, $m_\alpha^\pi = m_{\beta'}^{\phantom{\pi}}$ where
  $\beta'_j = \sum_{i\in \pi^{-1}(j)} \alpha_i$, and consequently
  \[
    \gamma([m_\alpha]^\pi) = \gamma([m_\alpha^\pi]) = \beta = \gamma(m_\alpha).
    \qedhere
  \]
\end{proof}

In the next subsection, we show that $\mu$ and $\gamma$ are in fact inverses to each other, and hence the previous lemma applies for all elements of $\hpol{S^1, P^2}$, i.e., $\gamma$ (and consequently $\mu$) preserves all minors.

\subsection{Monomial maps cover all possible maps up to homotopy}\label{sec:no_more}

To show that there are no other maps besides the monomial maps $m_\alpha$, $\alpha\in \affine^3$ up to equivariant homotopy, we will rely on Theorem~\ref{thm:obstruction} in the particular case of $X = T^n$ and $P^2$ playing the role of $K(\pi_2(P^2), 2)$, which asserts that there is a bijection $[T^n, P^2]_{\ints_3} \simeq H_{\ints_3}^2(T^n; \pi_2(P^2))$.
The rest of this section will be dedicated to precisely define the objects involved in Theorem~\ref{thm:obstruction} and compute the Bredon cohomology of the torus, i.e., to prove the following proposition.

\begin{restatable}{proposition}{cohomologytorus} \label{prop:cohomology_torus}
  For all $n, d\geq 1$,
  \[
    H_{\Z_3}^d\left(T^n; \pi_2(P^2)\right) \cong \Z_3^{\binom{n-1}{d-1}}.
  \]
\end{restatable}

\subsubsection{Equivariant cohomology: a primer}\label{sec:cohomology_def}

We introduce equivariant cohomology in a specific case that is sufficient for our purposes, and refer to Bredon \cite[Section 1]{Bre67}, or tom Dieck \cite[Section II.3]{Die87} for a detailed treatment.

As we have previously observed, if a CW-complex $X$ has a free cellular $\Z_3$-action, then there is an induced action on the group of $d$-dimensional integer chains $C_d(X)$ that commutes with the boundary maps.
The main insight of Bredon equivariant (co)homology is a shift in perspective: specifying an action on the chains is exactly the same as giving the chain group a structure of $\Lambda$-module for a suitable constructed ring (the integer group ring $\Z[\Z_3]$).\footnote{This holds in our specific context. In the general case, the Bredon cohomology has a richer structure than simply that of $\Lambda$-module.}

More explicitly, the ring $\Lambda$ is defined as the formal integer combinations of elements of $\Z_3 = \{1, \omega, \omega^2\}$ with the multiplication defined by setting $\omega^i \omega^j = \omega^{i+j}$ and extending bilinearly.
The key observation is that prescribing a $\Z_3$-action on an Abelian group $M$ is exactly the same as choosing a structure of $\Lambda$-module on $M$.
In particular, if $X$ is a CW-complex with a cellular $\Z_3$-action, for any $d\geq 0$, the corresponding chain group $C_d(X)$ can be seen as a $\Lambda$-module with multiplication given on a cell $\sigma$ by
\[
  (n_0+n_1\omega+n_2\omega^2) \sigma =  n_0 \sigma + n_1 (\omega\cdot \sigma) + n_2 (\omega^2\cdot \sigma)
\]
and extended linearly.

Since the all the boundary maps commute with the action, these are also maps of $\Lambda$-modules hence $C_\bullet(X)$ is a chain complex of $\Lambda$-modules.
We will denote this chain complex by $C^{\Z_3}_\bullet(X)$.
The homology associated to this chain complex is the \emph{equivariant homology} (or Bredon homology) of $X$, denoted by $H_\bullet^{\Z_3}(X)$.

At the same time, if we fix a $\Lambda$-module $N$, we can consider the equivariant cochain complex:
\[
  C_{\Z_3}^i \left(X; N\right) = \Hom_\Lambda\bigl(C^{\Z_3}_i(X), N\bigr)
\]
with the standard coboundary maps. The cohomology of this cochain complex is the Bredon cohomology, denoted by $H_{\Z_3}^\bullet(X; N)$.

\begin{lemma}
  \label{thm:free_module}
  If the action on $X$ is free and cellular, then $C_\bullet^{\Z_3}(X)$ is a chain complex of free $\Lambda$-module.
\end{lemma}

\begin{proof}
  For every orbit of $d$-cells in $X$ choose a representative $\sigma$. Then the map
  \begin{align*}
    C_d\left(\faktor{X}{\Z_3}; \Lambda\right) &\longrightarrow C^{\Z_3}_d(X)\\
    [\sigma] &\longmapsto \sigma
  \end{align*}
  is an isomorphism of $\Lambda$-modules; moreover, since the action is cellular, these isomorphisms commute with the boundary maps hence the two chain complexes are isomorphic.
\end{proof}

This implies that the functor $\Hom_\Lambda(C_d^{\Z_3}(X), {-})$ is exact for all $d\geq 0$. Therefore, if we have a short exact sequence of $\Lambda$-modules
\[
  0\rightarrow N \rightarrow M \rightarrow Q \rightarrow 0,
\]
there is a corresponding short exact sequence of cochain complexes
\[
  0\rightarrow C_{\Z_3}^\bullet(X; N) \rightarrow C_{\Z_3}^\bullet(X; M) \rightarrow C_{\Z_3}^\bullet(X; Q) \rightarrow 0,
\]
and thus a long exact sequence in cohomology
\[
  \dots\rightarrow H_{\Z_3}^{i-1}(X; Q) \rightarrow H_{\Z_3}^i(X; N) \rightarrow H_{\Z_3}^i(X; M) \rightarrow H_{\Z_3}^i(X; Q) \rightarrow \dots
\]

The key technical lemma we are going to use to compute the equivariant cohomology of the torus $T^n$ follows after we fix some notation.

Assume that $X$ is a CW complex with a free $\ints_3$-action.
Let $p\colon X\to {X} / {\Z_3}$ be the projection map.
Let $I = \Lambda(1 + \omega + \omega^2)$ be the ideal in $\Lambda$ generated by the element $1 + \omega + \omega^2$, and let $\iota\colon I\hookrightarrow \Lambda$ denote the inclusion.
Let $Z$ be the $\Lambda$-module defined as $Z = \Z$ with the trivial multiplication by $\omega$ and let $Z_{\text{cyc}}$ the $\Lambda$-module $Z_{\text{cyc}} = \Z\oplus \Z\oplus \Z$ where the multiplication by $\omega$ is a cyclic shift, i.e., $\omega\cdot(n_0, n_1, n_2) = (n_2, n_0, n_1)$.
Furthermore, we have the following homomorphisms:
\begin{enumerate}
  \item an inclusion of $\Lambda$-modules $d\colon Z\hookrightarrow Z_\text{cyc}$ defined as $d(n) = (n, n, n)$;
  \item an isomorphism of $\Lambda$-modules $\phi_1\colon Z \to I$ defined as $\phi_1(n) = n + n \omega + n \omega^2$;
  \item an isomorphism of $\Lambda$-modules $\phi_2\colon Z_\text{cyc} \to \Lambda$ defined by $\phi_2(n_0, n_1, n_2) = n_0 + n_1 \omega + n_2 \omega^2$;
  \item an isomorphism of Abelian groups $h_1\colon C^\bullet(\faktor{X}{\Z_3}) \to C^\bullet_{\Z_3}(X; Z)$ defined by $h_1(\alpha) = \left(\sigma \mapsto \alpha(p(\sigma))\right)$; and finally
  \item an isomorphism of Abelian groups $h_2\colon C^\bullet\left(X\right) \to C^\bullet_{\Z_3}(X; Z_\text{cyc})$ defined by
    \[
      h_2 (\alpha) = \left(\sigma \mapsto (\alpha(\sigma), \alpha(\omega\cdot \sigma), \alpha(\omega^2\cdot \sigma))\right).
    \]
\end{enumerate}
All of these claims are straightforward to check.

\begin{lemma} \label{lem:key_lemma_cohomology}
  Let $X$ be a CW complex with a free $\ints_3$-action, then the following diagram commutes
  \[
    \begin{tikzcd}
      C^\bullet(\faktor{X}{\Z_3}) \ar[d, "p^*"] \ar[r, "h_1"] & C^\bullet_{\Z_3}\left(X;Z \right) \ar[d, "{-}\circ d"]\ar[r, "{-}\circ \phi_1"] & C^\bullet_{\Z_3}\left(X; I\right) \ar[d, "{-}\circ \iota"]\\
      C^\bullet\left(X\right) \ar[r, "h_2"] &C^\bullet_{\Z_3}\left(X; Z_\text{cyc}\right) \ar[r, "{-}\circ \phi_2"] & C^\bullet_{\Z_3}\left(X; \Lambda\right)
    \end{tikzcd}
  \]
\end{lemma}

\begin{proof}
  The proof is a straightforward unravelling of definitions.
\end{proof}

\subsubsection{Computing the cohomology of the torus}
  \label{sec:computing_cohomology_torus}

We are now going to specify the previous construction to our case, that is $X = T^n$ with the diagonal action (that is cellular for the product CW-structure on the torus) and coefficients $\pi_2(P^2)$ (that is a $\Lambda$-module with the induced action).

The first ingredient is to explicitly study the structure of $\Lambda$-module on $\pi_2(P^2)$.
  \begin{lemma}
    The second homotopy group $\pi_2(P^2)$ is isomorphic as a $\Lambda$-module to $M=\Z\oplus\Z$ where the multiplication by $\omega$ is given by the matrix
    \[
      A=\begin{pmatrix}
        0 & -1 \\
        1 & -1
      \end{pmatrix}
    \]
  \end{lemma}

  \begin{proof}
    Since $P^2$ is simply connected, by Hurewicz theorem \cite[see Section 4.2]{Hat02}, there is an isomorphism between $\pi_2(P^2)$ and $H_2(P^2)$ that is an isomorphism of $\Lambda$-modules.

    By direct computation, it is easy to show that the chains $\beta_i = d_i - d_{i+1}\in C_2(P^2)$ are cycles and any two out of $\{[\beta_i]\}_{0\leq i\leq 2}$ form a basis for $H_2(P^2)\cong \Z\oplus \Z$. Fix $[\beta_0]$ and $[\beta_1]$, then the multiplication by $\omega$ is
    \begin{align*}
      \omega\cdot[\beta_0] = [\omega\cdot\beta_0] & = [\beta_1]\\
      \omega\cdot[\beta_1] = [\omega\cdot\beta_1] & = [\beta_2] = -[\beta_0] - [\beta_1].
    \end{align*}
    as claimed.
  \end{proof}

The module $M$ has some very useful properties that will allow us to use Lemma~\ref{lem:key_lemma_cohomology}. In particular, we have the following lemma.

  \begin{lemma}\label{thm:coefficient_structure}
    The module $M$ is generated by a single element and $\Ann(M) = I$ the ideal generated by $1+\omega+\omega^2$.
  \end{lemma}

  \begin{proof}
    $M$ is generated over $\Lambda$ by $m=\begin{pmatrix} 1\\0 \end{pmatrix}$. As a consequence, $\lambda\in \Ann(M)$ if and only if $\lambda m = 0$; hence, if $\lambda = n_0 + n_1 \omega + n_2 \omega^2$, then
    \[
      \lambda m = n_0 \begin{pmatrix} 1\\0 \end{pmatrix} + n_1 \begin{pmatrix} 0\\1\end{pmatrix} + n_2\begin{pmatrix} -1 \\-1 \end{pmatrix} = \begin{pmatrix}
        n_0 - n_2\\
        n_1 - n_2
      \end{pmatrix}.
    \]
    Therefore, $\lambda m = 0$ if and only if $n_0 = n_1 = n_2$ if and only if $\lambda \in I$.
  \end{proof}

The last ingredient needed is to compute what the projection map $p^*$ does on the level of cohomology.
While it is possible to compute $p^*$ directly, it is easier to change a little the action on the torus to simplify the calculations:

\begin{lemma}
  \label{thm:simple_torus}
  Let $X$ be the torus $T^n\subseteq \C^n$ with the diagonal $\Z_3$-action given by the multiplication with $\omega = e^{2\pi i / 3}$ (i.e., $\omega\cdot(z_1, \dots, z_n) = (\omega z_1, \dots, \omega z_n))$.
  Let $Y$ be the same torus but with $\Z_3$ acting only on the first coordinate (i.e., $\omega\cdot(z_1, \dots, z_n) = (\omega z_1, z_2, \dots, z_n)$).
  Then there is a $\Z_3$-equivariant homeomorphism $X\to Y$.
\end{lemma}

\begin{proof}
  The maps $h\colon X \to Y$ and $h'\colon Y \to X$ defined by
  \begin{align*}
    h\colon (z_1, \dots, z_n) &\mapsto (z_1, z_1^{-1} z_2, \dots, z_1^{-1} z_n) \\
    h'\colon (z_1, \dots, z_n) &\mapsto (z_1, z_1 z_2, \dots, z_1 z_n)
  \end{align*}
  are clearly continuous and mutually inverse.
  We will show that $h$ preserve the actions involved, and hence that $h$ is an equivariant homeomorphism:
  \[
    h(\omega\cdot_X (z_1, \dots, z_n)) = h(\omega z_1, \dots, \omega z_n) = (\omega z_1, z_1^{-1} z_2, \dots, z_n z_1^{-1}) = \omega\cdot_Y h(z_1, \dots, z_n) \qedhere
  \]
\end{proof}

\begin{remark}
  If we see a torus $T^n$ as a quotient of $\R^n$ over the standard lattice, Lemma~\ref{thm:simple_torus} shows that factoring out the action is the same as factoring out the lattice generated by $\{\frac{1}{3} e_1, e_2, \dots, e_n\}$.
  This means that, topologically, the quotient is still a torus.
\end{remark}

For all future homological calculations, we can assume that $\Z_3$ acts by multiplying by $\omega$ on the first coordinate.
Using this simplified action on the torus it is much easier to compute the quotient map $p^*\colon H^\bullet(T^n/\Z_3)\to H^\bullet(T^n)$.
To achieve this objective, it is necessary to fix a basis for the cohomology of the torus.
The ideal choice would be a basis that is ``easy'' to evaluate on homology classes in order to compute easily the image of $p^*$.
In the case of the torus, a direct application of the universal coefficient theorem (\cite[see Section 3.1]{Hat02}) show that homology and cohomology in dimension $1$ are dual to each other; hence we can choose as basis for the first cohomology group the dual of a clever basis for the first homology group.
In particular, let $\{x_i\}$ be the basis for $H_1(T^n)$ corresponding to the coordinate loops in $T^n$ and denote by $\{x^i\}$ the dual basis in $H^1(T^n)\cong \Hom\left(H_1(T^n), \Z\right)$; define analogously $\{q_i\}$ basis for $H_1({T^n}/{\Z_3})$ and $\{q^i\}$ for $H^1({T^n}/{\Z_3})$.
Then
\[
  p^*_1(q^i) =
  \begin{cases}
    3x^1 & \text{if $i=1$}\\
    x^i & \text{otherwise.}
  \end{cases}
\]

The ring structure on cohomology (see \cite[Section 3.2]{Hat02}) of the torus allows us to build a convenient basis for all the other cohomology groups out of $\{x^i\}$.
In fact, elements of the form $x^I = x^{i_1}\smile \dots\smile x^{i_d}$, where $I = (i_1, \dots, i_d)$ and $i_1 < \dots < i_d$, form a basis for $H^d(T^n)$. Let $q^I$ denote the analogous basis for $H^d({T^n}/{\Z_3})$.
Since $p^*$ is a ring map, it commutes with the cup product, hence it can be explicitly computed on such a basis. We have that, for all $d$,
\[
  p^*_d(q^I) = p^*_1(q^{i_1})\smile \dots \smile p^*_1(q^{i_d})=
  \begin{cases}
    3x^I & \mbox{ if }i_1 = 1\\
     x^I & \mbox{ else.}
  \end{cases}
\]
In particular, $p^*_d$ is injective for all $d\geq 0$ and, in this choice of basis, $p^*$ is the diagonal matrix with $\binom{n-1}{d-1}$ 3's and $\binom {n-1}{d}$ 1's on the diagonal.

We are finally ready to compute the equivariant cohomology group of the torus $T^n$ and prove Proposition~\ref{prop:cohomology_torus} which claims that
\(
  H_{\Z_3}^d\left(T^n; \pi_2(P^2)\right) \cong \Z_3^{\binom{n-1}{d-1}}
\).

\begin{proof}[Proof of Proposition~\ref{prop:cohomology_torus}]
  Fix $n\geq 2$.
  By Lemma~\ref{thm:coefficient_structure}, we have a short exact sequence
  \[
    \begin{tikzcd}[cramped, sep=small]
      0\ar[r] & I \ar[r] & \Lambda \ar[r]& M\ar[r]&0
    \end{tikzcd}
  \]
  which induces short exact sequence of cochain complexes
  \[
    \begin{tikzcd}[cramped, sep=small]
      0\ar[r] & C^\bullet_{\Z_3}\left(T^n; I\right)\ar[r] & C_{\Z_3}^\bullet\left(T^n; \Lambda\right) \ar[r]& C^\bullet_{\Z_3}\left(T^n; M\right)\ar[r]&0
    \end{tikzcd}
  \]
  Using Lemma~\ref{lem:key_lemma_cohomology}, we get that the following short sequence is also exact
  \[
  \begin{tikzcd}[cramped, sep=small]
    0\ar[r] & C^\bullet\left(\faktor{T^n}{\Z_3}\right)\ar[r, "p^*"] & C^\bullet\left(T^n\right) \ar[r]& C^\bullet_{\Z_3}\left(T^n; M\right)\ar[r]&0
  \end{tikzcd}
  \]
  This short exact sequence induces the following long exact sequence in cohomology
  \[
    \begin{tikzcd}[cramped, sep=small]
      \dots \ar[r] & H^d\left(\faktor{T^n}{\Z_3}\right) \ar[r, "p^*_d"] & H^d\left(T^n\right) \ar[r] & H_{\Z_3}^d\left(T^n; M\right) \ar[r]& H^{d+1}\left(\faktor{T^n}{\Z_3}\right) \ar[r, "p^*_{d+1}"] & H^{d+1}\left(T^n\right) \ar[r] & \dots
    \end{tikzcd}
  \]
  Since $p^*_d$ is injective for any $d\geq 1$, by exactness we have that
  \(
    H_{\Z_3}^d\left(T^n; M\right) \cong \coker p^*_d
  \).
  Furthermore, $\coker p^*_d = \ints^{\binom nd} / \operatorname{im} p^*_d \simeq \ints_3^{\binom{n-1}{d-1}}$ which yields the desired.
\end{proof}

\subsection{The minion isomorphism}

Putting everything together, we can now provide the required isomorphism of minions $\affine_3$ and $\hpol{S^1, P^2}$. Recall the definition of maps $\mu\colon \affine_3 \to \hpol{S^1, P^2}$ which maps $\alpha$ to the homotopy class of the monomial map $m_\alpha$, and $\gamma\colon \hpol{S^1, P^2} \to \affine_3$ which maps a homomotopy class of a map $f$ to the sequence of its degrees.

\begin{proof}[Proof of Lemma~\ref{thm:affine_minion}]
  We show that $\mu$ and $\gamma$ are mutually inverse minion homomorphisms. Recall that $\mu$ is injective by Lemma~\ref{lem:all_distinct}. We first show that it is surjective. Indeed, we have that, for each $n \geq 1$,
  \[
    \card{ [T^n, P^2]_{\Z_3} } = 3^{n-1}
  \]
  by a combination of Proposition~\ref{prop:cohomology_torus} and Theorem~\ref{thm:obstruction} since $H^2(T^n, \pi_2(P^2)) \simeq \ints_3^{n-1}$, and the latter has $3^{n-1}$ elements.

  The above, in particular, means that $\mu_n \colon \affine_3^{(n)} \to \hpol[n]{S^1, P^2}$ is onto, and hence a bijection.
  Consequently, $\gamma$ is a minion homomorphism by Lemma~\ref{lem:gamma_is_homomorphism} since every class in $[T^n, P^2]_{\ints_3}$ contains a monomial map.
  Observe that $\mu\circ \gamma$ is the identity map by Lemma~\ref{lem:monomial_degrees}, which implies that $\gamma$ is the inverse of the bijection $\mu$.

  Finally, an inverse of a bijective minion homomorphism $\gamma$ is a minion homomorphism since, for all $\alpha \in \affine_3$,
  \[
    \mu(\alpha^\pi)
    = \mu(\gamma(\mu(\alpha))^\pi)
    = \mu(\gamma(\mu(\alpha)^\pi))
    = \mu(\alpha)^\pi.
  \]
  This concludes that $\mu$ and $\gamma$ are the required minion isomorphisms.
\end{proof}

\section{A minion homomorphism}
  \label{app:categories}

The goal of this appendix is to provide minion homomorphisms
\[
  \Pol(\LO_3, \LO_4) \to \hpol{\homLO3, \homLO4} \to \hpol{S^1, P^2}
\]
using the results of previous appendices. In particular, the second of the above homomorphism uses the equivariant maps $S^1 \to \homLO3$ and $\homLO4 \to P^2$.
Let us note that similar minion homomorphisms have been constructed in \cite[Lemma 3.22]{KOWZ23}, and the proof can be altered to construct the homomorphisms, we require here. Rather than doing that, we provide two general lemmata that cover both of these uses as well as any possible future uses.
We formulate and prove these two lemmata in the language of category theory which is well suited. We do not assume deep knowledge of category theory and recall the necessary core notions.

\paragraph*{Categories, functors, and natural transformations}

Instead of giving precise definitions, which can be found in any category theory textbook, we will give examples of the basic notions of category theory.

Let us recall that a \emph{category} $\cat$ is a class of \emph{objects} (e.g., relational structures of a given signature, topological spaces, groups, sets, etc.) together with a set $\hom(A, B)$ of \emph{morphisms} for each pair of objects $A, B \in \cat$ (e.g., homomorphisms from $A$ to $B$, etc.).
Morphisms can be composed in a natural way, and for each object $A$ there is a (two-sided) identity morphism $1_A \in \hom(A, A)$.
Also recall that an \emph{isomorphism} is an invertible morphism, i.e., $f\in \hom(A, B)$ is an isomorphism if there is $g \in \hom(B, A)$ such that $f\circ g = 1_A$ and $g\circ f = 1_B$.

\begin{example}
  The most basic category is the category of sets: objects are sets, and morphisms are maps between the two given sets. We denote this category by $\set$.
  Relational structures of a fixed signature $\Sigma$ together with homomorphisms form a category which we denote by $\rels_\Sigma$ or usually just $\rels$ when the signature is clear from the context. The categorical notion of isomorphism coincides with structural isomorphism.
\end{example}

\begin{example}
  The usual structure of a category on topological spaces has continuous maps as morphisms, and homeomorphisms as isomorphisms.
  Nevertheless, we will work with a different structure of a category on topological spaces. Namely, the category $\htop_G$ whose objects are topological spaces with an action of a fixed group $G$, and morphisms are \emph{equivariant homotopy classes of maps}, i.e., $\hom(X, Y) = [X, Y]_G$. The composition is defined by $[f] \circ [g] = [f \circ g]$.\footnote{Do not confuse this category with the \emph{homotopy category} which identifies \emph{weakly} homotopy equivalent spaces.}
  We will write just $\htop$ instead of $\htop_{\ints_3}$.

  Isomorphisms in this category coincide with equivariant homotopy equivalences: If $f\colon X \to Y$ and $g\colon Y \to X$ witness a homotopy equivalence, then both $[f]$ and $[g]$ are isomorphisms in $\htop_G$, since we have $f\circ g \sim 1_X$ and $g\circ f \sim 1_Y$, and consequently $[f] \circ [g] = 1_X$ and $[g] \circ [f] = 1_Y$.
  Loosely speaking, this category is obtained from the category of topological spaces with a $G$-action by identifying equivariantly homotopic spaces.
\end{example}

\begin{example}
  Another useful category is the category whose objects are non-negative integers, and morphisms from $n$ to $m$ are maps $[n] \to [m]$, i.e., $\hom(n, m) = \{f\colon [n] \to [m]\}$. We denote this subcategory of the category of sets by $\nat$.
  Note that $\nat$ contains all important information about finite sets since it has an object for each isomorphism class (cardinality) of finite sets.
\end{example}

A \emph{functor} is a natural notion of a morphism between categories. A functor $\minion F\colon \cat_1 \to \cat_2$ is a mapping between objects and morphisms that preserves composition and identity morphisms. More specifically, we require that $\minion F(A) \in \cat_2$ for each $A \in \cat_1$, and $\minion F(f) \in \hom(\minion F(A), \minion F(B))$ for each $f\in \hom(A, B)$ where $A, B \in \cat_1$.

\begin{example}
  We will view the assignment $\minion B\colon \rel A \mapsto \homrel A$ as a functor $\rels \to \htop$. This functor is defined on morphism by mapping a homomorphism $f\colon \rel A \to \rel B$ to the class $[f_*]$ where $f_*\colon \homrel A \to \homrel B$ is the induced continuous map.
\end{example}

\begin{example}
  An abstract minion is a functor $\minion M\colon \nat \to \set$ such that $\minion M(n) \neq \emptyset$ for all $n \neq 0$.
\end{example}

Finally, \emph{natural transformations} are morphisms between functors.

\begin{example}
  A natural transformation between two minions coincides with the notion of minion homomorphism: a natural transformation $\eta\colon \minion M \to \minion N$ is a collection of maps $\eta_n$ where $n$ ranges through objects of $\nat$ such that, for each $n, m \in \nat$ and $\pi \in \hom(n, m)$, the following square commutes.
  \[\begin{tikzcd}
    \minion M(n) \arrow[d, "\eta_n"'] \arrow[r, "\pi^\minion M"] & \minion M(m) \arrow[d, "\eta_m"] \\
    \minion N(n) \arrow[r, "\pi^\minion N"']                     & \minion N(m)
  \end{tikzcd}\]
  It is easy to observe that this square commutes if and only if $\eta$ preserves minors.
\end{example}

In general, we will say that a map $\eta_A$ that depends on an object $A$ of some category is \emph{natural in $A$} if some square, obtained by varying $A$ by a morphism $a\colon A \to A'$, commutes. Which square commutes is usually clear from the context.

\paragraph*{Products and polymorphisms}

A polymorphism is a homomorphism from a power, and hence if we want to define polymorphisms in a category, we need to have a notion of a power, or more generally \emph{product}. Products in an arbitrary category are defined by how they interact with other objects. More precisely, we can observe that products of sets, structures, or topological spaces there is a bijection
\[
  \hom(A, B_1) \times \dots \times \hom(A, B_n) \simeq \hom(A, B_1 \times \dots \times B_n)
\]
for each object $A$, i.e., in order to define a morphism from an object $A$ to a product of objects $B_1$, \dots, $B_n$, it is enough to specify an $n$-tuple of morphisms from $A$ to each $B_i$. (In order to define products precisely, we would also require that the above bijection is natural in $A$.) Formally, the products are defined as follows.

\begin{definition}
  Let $B_1, \dots, B_n \in \cat$. We say that an object $P$ together with morphisms $p_i \in \hom(P, B_i)$ for $i \in [n]$ is the \emph{product} of $B_1$, \dots, $B_n$ if, for each $A \in \cat$ and morphisms $a_i \in \hom(A, B_i)$, there is a unique morphism $a\in \hom(A, P)$ such that $a_i = p_i \circ a$. The maps $p_i$ are called \emph{projections}.
\end{definition}

Let us note that the projections can be obtained from the above bijection by substituting $P$ for $A$ and considering the identity map on $P$ which is assigned a tuple of morphisms $(p_1, \dots, p_n) \in \hom(P, B_1) \times \dots \times \hom(P, B_n)$.
Observe that products are well-defined only up to isomorphisms. Indeed, if $(P, p_i)$ and $(P', p'_i)$ are both products of $B_1$, \dots, $B_n$, then there are morphisms $f\in \hom(P, P')$ and $g\in \hom(P', P)$ such that $p_i' = p_i \circ f$ and $p_i = p_i' \circ g$. Since $1_P$ is the unique morphisms such that $p_i = p_i \circ 1_P$, we get that $g\circ f = 1_P$. Analogously, $f\circ g = 1_P'$.

\begin{example}
  Products do not always exists. Nevertheless, all the categories we work with in this paper ($\nat$, $\set$, $\rels$, and $\htop_G$) have all (finite) products which coincide with the usual definitions.

  Let us comment on products in $\htop_G$ which are a bit subtle. The product of $G$-spaces $X_1$, \dots, $X_n$ is the space $X_1 \times \dots \times X_n$ with the usual product topology and the coordinatewise (diagonal) action of $G$. The projections are then the homotopy classes of usual projections.
  Let us show the universal property of the product. Let $Y$ be a $G$-space, and let $f_i\colon Y \to X_i$ be continuous maps. We claim that there is a map $f\colon Y \to X_1 \times \dots \times X_n$ such that $p_i \circ f$ and $f_i$ are equivariantly homotopic, and this map is unique up to equivariant homotopy. The existence is straightforward since we can define $f$ in the usual way:
  \[
    f(y_1, \dots, y_n) = (f_1(y_1), \dots, f_n(y_n))
  \]
  It is easy to check that this map is indeed $G$-equivariant since the $f_i$'s are, moreover, we have $p_i \circ f = f_i$. To show that $f$ is unique up to equivariant homotopy, let $g$ be such that $p_i \circ g$ and $f_i$ are homotopic for all $i$. This means we have homotopies $h_i\colon [0, 1] \times Y \to X_i$ between $f_i$ and $p_i \circ g$ for each $i$. By the universal property of product in topological spaces, these define a homotopy $h\colon [0, 1] \times Y \to X_1 \times \dots \times X_n$, i.e., we define $h$ by the same prescription as $f$ using $h_i$ instead of $f_i$. This $h$ is the required homotopy between $f$ and $g$.
\end{example}

We can now define the notion of a polymorphism in any category with products.

\begin{definition}
  Let $\cat$ be a category with finite products, $A$, $B \in \cat$ be two objects, and $n \geq 0$. We define a \emph{polymorphism} from $A$ to $B$ of arity $n$ to be any element $f \in \hom(A^n, B)$ where $A^n$ is the $n$-fold power of $A$.
\end{definition}

In order to define the \emph{polymorphism minion} $\Pol(A, B)$ in such a category $\cat$, we need a functor from $\nat$ to $\set$ that assigns to each $n$, the set $\hom(A^n, B)$. An easy way to observe that such a functor can be always defined is to decompose it as two contravariant (i.e., arrow-reversing) functors $\nat \to \cat$ and $\cat \to \set$: the first of which assigns to $n$ the $n$-fold power $A^n$ of $A$, and the second of which assigns to this $A^n$ the set $\hom(A^n, B)$.

Let us first observe that the assignment $n \mapsto A^n$ can be extended to a functor $A^{-}\colon \nat \to \cat$ for each $A \in \cat$. This can be relatively easily observed in all concrete cases, e.g., if $A$ is a structure $\rel A$, then, for each $\pi\colon [n] \to [m]$, the mapping $\rel A^\pi\colon A^m \to A^n$ is defined by $a \mapsto a \circ \pi$ which is clearly a homomorphism. To give a general definition, we use the universal property of products.
Let $\pi\colon [n] \to [m]$ be a mapping. We want to define a homomorphism $p^A_\pi\colon A^m \to A^n$. Using the definition of the $n$th power, it is enough to give an $n$-tuple of homomorphisms $a_i\colon A^m \to A$, $i \in [n]$. We let $a_i = p_{\pi(i)}$ where $p_j$ denote projections of the $m$th power of $A$. Finally, it is straightforward to check that the assignment $\minion A(\pi) = p_\pi^A$ preserves compositions and identities.

The second of these assignments is known as \emph{contravariant hom-functor} $\hom({-}, B)$. It maps an object $A \in \cat$ to the set $\hom(A, B)$, and a morphism $f\in \hom(A, A')$ to the mapping ${-}\circ f\colon \hom(A', B) \to \hom(A, B)$ defined by $g\mapsto g\circ f$.

\begin{remark}
  If $\cat = \set$, then the functor $A^{-}$ can be also described as a restriction of the contravariant hom-functor $\hom({-}, A)$ to $\nat$ viewed as a subcategory of $\set$.
\end{remark}

\begin{definition}
  Let $\cat$ be a category with finite products, and $A, B \in \cat$ be such that $\hom(A, B)$ is non-empty. We define the \emph{polymorphism minion} $\pol_\cat(A, B)$ as the composition of functors $A^{-}$ and $\hom({-}, B)$. We will omit the index $\cat$ whenever the category is clear from the context.
\end{definition}

\begin{example}
  The polymorphism minion in the category $\htop$ coincides with the minion of homotopy classes of polymorphisms, i.e., for all spaces $X, Y \in \htop$, we have $\pol_\htop(X, Y) = \hpol{X, Y}$.
\end{example}

Note that the definition of this polymorphism minion does not depend (up to natural equivalence) on which of the realisation of the power functor $A^{-}$ we take.
Finally, we recall the categorical definition of preserving products, which ensures in particular that $\minion F(A)^{-}$ is naturally equivalent to the composition of $A^{-}$ and $\minion F$.

\begin{definition}
  We say that a functor $\minion F\colon \cat_1 \to \cat_2$ \emph{preserves finite products} if, for each $A_1, \dots, A_n \in \cat_1$ and their product $(P, p_i)$, $(\minion F(P), \minion F(p_i))$ is a product of $\minion F(A_1), \dots, \minion F(A_n)$.
\end{definition}

\begin{lemma}
  If a functor $\minion F$ preserves products, then $\minion F \circ A^{-}$ and $\minion F(A)^{-}$ are naturally equivalent.
\end{lemma}

\begin{proof}
  We define a natural equivalence $\eta\colon \minion F\circ A^{-} \to \minion F(A)^{-}$ by components as $\eta_n\colon \minion F(A^n) \to \minion F(A)^n$ is defined as the map given by the $n$-tuple $\minion F(p_1), \dots, \minion F (p_n)\colon \minion F(A^n) \to \minion F(A)$. Since $\eta_n$ commutes with the projections of the two products, it is an isomorphism by the same argument as we used in showing that product is unique up to isomorphisms.
  It is also straightforward to check that $\eta$ is natural using the universal property of the product.
\end{proof}

\subsection{Two general lemmata}

Given the definitions above, the first lemma is rather trivial. It claims that we can transfer polymorphisms through any functor that preserves products. It was mentioned in the context of promise CSPs in \cite{WZ20}, although in different flavours it can be tracked down to Lawvere theories.

\begin{lemma} \label{lem:wrochna}
  If a functor $\minion F\colon \cat_1 \to \cat_2$ preserves products, then there is a minion homomorphism
  \[
    \xi\colon \Pol(A, B) \to \Pol(\minion F(A), \minion F(B))
  \]
  for all $A, B \in \cat_1$ such that $\hom(A, B) \neq \emptyset$.
\end{lemma}

\begin{proof}
  By definition $\Pol(\minion F(A), \minion F(B))$ is the composition of $\minion F(A)^{-}$ and $\hom({-}, \minion F(B))$. Since the first functor is naturally equivalent to $\minion F\circ A^{-}$, we may assume they are equal. We can then define a natural transformation $\xi\colon \hom(A^{-}, B) \to \hom(\minion F(A^{-}), \minion F(B))$ by
  \[
    \xi_n(f) = \minion F(f).
  \]
  To show $\xi$ preserves minors, observe that $f^\pi$ is defined as $f\circ p_\pi^A$, and consequently
  \(
    \minion F(f^\pi)
    = \minion F(f \circ p_\pi^A)
    = \minion F(f) \circ \minion F(p_\pi^A)
    = \minion F(f) \circ p_\pi^{\minion F(A)}
    = \minion F(f)^\pi
  \).
\end{proof}

The second lemma is a direct generalisation of \cite[Lemma~4.8(1)]{BBKO21}, and the proof is analogous.

\begin{lemma} \label{lem:relaxation}
  Let $A, A', B, B' \in \cat$ be such that $\hom(A, B)$ is non-empty. Then every pair $a\in \hom(A', A)$ and $b\in \hom(B, B')$ induces a minion homomorphism
  \[
    \xi\colon \Pol(A, B) \to \Pol(A', B').
  \]
\end{lemma}

\begin{proof}
First, observe that $a$ induces a natural transformation $A^{-} \to (A')^{-}$, which can be shown by using the definition of products. We will denote its components by $a^n \in \hom (A^n, (A')^n)$.
The minion homomorphism is defined by $\xi_n(f) = b \circ f \circ a^n$. To show that $\xi$ preserves minors, observe that the following diagram commutes.
  \[
    \begin{tikzcd}
    (A')^n \arrow[d, "p_\pi^{A'}"'] \arrow[r, "a^n"] & A^n \arrow[d, "p_\pi^A"'] \arrow[rd, "f^\pi"] &                  &    \\
    (A')^m \arrow[r, "a^m"]                          & A^m \arrow[r, "f"']                           & B \arrow[r, "b"] & B'
    \end{tikzcd}
  \]
  Consequently, we have that
  \(
    \xi(f^\pi)
    = b \circ f^\pi \circ a^n
    = (b \circ f \circ a^m) \circ p_\pi^{A'}
    = \xi_m(f)^\pi
  \).
\end{proof}

\subsection{Homomorphism complexes preserve products}

We prove that the functor $\homrel{{-}}$ preserves products in the categorical sense. Essentially, it follows by the same argument as the well-known fact that the product of homomorphisms complexes is homotopically equivalent to the homomorphism complex of the product (see, e.g., \cite[Section 18.4.2]{Koz08}).
We provide a bit more detailed proof to show that the homotopy equivalence can be taken equivariant, and that the products are preserved in the categorical sense which is a slightly stronger statement. Let us first prove that there is an equivariant homotopy equivalence. To provide this equivalence, we use the following technical lemma.

\begin{lemma}
  Let $\complex K$ be an order complex of a poset. If $\alpha, \beta\colon V(\complex K) \to V(\complex K)$ are monotone maps such that $\alpha(x) \leq \beta(x)$ for each $x$, then their linear extensions are homotopic.
\end{lemma}

\begin{proof}
  Observe that the map $h\colon \{0, 1\} \times V(\complex K) \to V(\complex K)$ defined by $h(0, x) = \alpha(x)$ and $h(1, x) = \beta(x)$ is monotone, and that the complex of the product poset is homeomorphic to $[0, 1] \times \complex K$. Consequently, $h$ induces a homotopy between $\alpha$ and $\beta$ by linear extension.
\end{proof}

\begin{lemma}
  For each $G$-structure $\rel A$, and structures $\rel B_1$, \dots, $\rel B_n$, there is a $G$-equivariant homotopy equivalence
  \[
    \Hom(\rel A, \rel B_1) \times \dots \times \Hom(\rel A, \rel B_n)
    \sim_G
    \Hom(\rel A, \rel B_1 \times \dots \times \rel B_n).
  \]
\end{lemma}

\begin{proof}
  For each functor $\minion F$, we always have a morphism
  \[
   f\in \hom(\minion F(\rel B_1\times \dots \times \rel B_n),
             \minion F(\rel B_1) \times \dots \times \minion F(\rel B_n))
  \]
  defined by the universal property of the product applied to the maps $\minion F(p_i)$. In the case $\minion F = \homrel {{-}}$, this morphism is represented by the continuous map
  \[
    \alpha :
    \Hom(\rel A, \rel B_1 \times \dots \times \rel B_n)
    \to \Hom(\rel A, \rel B_1) \times \dots \times \Hom(\rel A, \rel B_n)
  \]
  defined by $\alpha(m) = (p_1 \circ m, \dots, p_n \circ m)$ for each $m\in \mhom(\rel A, \rel B_1 \times \dots \times \rel B_n)$ and extending linearly. It is straightforward to check that $\alpha$ is an equivariant continuous map. We prove that it has a homotopy inverse $\beta$ defined on vertices by
  \[
    \beta(m_1, \dots, m_n)\colon x \mapsto m_1(x) \times \cdots \times m_n(x),
  \]
  and extending linearly. Clearly, $\beta$ is $G$-equivariant and $\alpha\circ \beta = 1$.

  We show that $\beta\circ \alpha$ is equivariantly homotopic to the identity. It is easy to observe that $\beta(\alpha(m)) \geq m$ for each $m\in \mhom(\rel A, \rel B_1 \times \dots \times \rel B_n)$. Since both $\alpha$ and $\beta$ are monotone and equivariant, so is the composition, and the inequality gives an equivariant homotopy between $\beta \circ \alpha$ and the identity by the above lemma.
\end{proof}

The proof above immediately gives that $\homrel{{-}}$ preserves products. This is since we showed that $[\alpha]$, which commutes with projections, is an isomorphism.

\begin{corollary} \label{cor:hom}
  For each $G$-structure $\rel A$, the functor $\Hom(\rel A, {-})\colon \rels \to \htop_G$ preserves products.
\end{corollary}

\subsection{The minion homomorphism \texorpdfstring{$\pol(\LO_3, \LO_4) \to \affine_3$}{}}

We can now conclude the proof of the lemma claiming that there is a minion homomorphism $\pol(\LO_3, \LO_4) \to \affine_3$.

\begin{proof}[Proof of Lemma~\ref{lem:minion-homomorphism}]
We have the following minion homomorphisms
\[
  \pol(\LO_3, \LO_4) \to \hpol{\homLO3, \homLO4} \to \hpol{S^1, P^2}.
\]
Where the first arrow follows from a combination of Corollary~\ref{cor:hom}, and Lemma~\ref{lem:wrochna}, and the second arrow follows from Lemmata~\ref{lem:relaxation}, \ref{lem:homlo3}, and \ref{lem:homlo4}. The last minion is isomorphic to $\affine_3$ by Lemma~\ref{thm:affine_minion}, which yields the required.
\end{proof}

\bibliographystyle{plainurl}

\end{document}